\documentclass[acmsmall]{acmart-arxiv}

\setcopyright{none}

\usepackage{booktabs}   
\usepackage{subcaption} 
\usepackage{amsmath,amssymb,amsthm,mathtools,stmaryrd}
\usepackage{proof}
\usepackage[all]{xy}
\usepackage{url}
\usepackage{enumitem}
\usepackage{hyperref}
\usepackage{thmtools}

\usepackage{xcolor}

\renewcommand\thmcontinues[1]{continued}

\newcommand{\simTr}{\sim^{\mathsf{Tr}}}

\renewcommand{\U}{\mathsf{U}}
\renewcommand{\S}{\mathsf{S}}
\renewcommand{\b}{\mathsf{b}}
\newcommand{\pa}{\mathsf{p}}
\newcommand{\id}{\mathsf{id}}
\newcommand{\comp}{\mathsf{comp}}

\newcommand{\isProp}{\mathsf{isProp}}
\newcommand{\isSet}{\mathsf{isSet}}
\newcommand{\isContr}{\mathsf{isContr}}
\newcommand{\rest}[1]{\overline{#1}}

\newcommand{\Set}{\mathsf{Set}}
\newcommand{\Prop}{\mathsf{Prop}}
\newcommand{\N}{\mathbb{N}}
\newcommand{\Nat}{\mathbb{N}}

\newcommand{\fix}{\mathsf{fix}}
\newcommand{\dfix}{\mathsf{dfix}}
\newcommand{\pfix}{\mathsf{pfix}}

\newcommand{\funext}{\mathsf{funext}}

\newcommand{\next}{\mathsf{next}}

\newcommand{\nuFg}{\nu F^{\mathsf{g}}}
\newcommand{\unfold}{\mathsf{unfold}}
\newcommand{\isBisim}[1]{\mathsf{is}\mbox{-} #1 \mbox{-}\mathsf{Bisim}}
\newcommand{\isTrBisim}[1]{\mathsf{is}\mbox{-} #1 \mbox{-}\mathsf{TrBisim}}
\newcommand{\isLTSBisim}{\mathsf{isGLTSBisim}}

\newcommand{\I}{\mathbb{I}}
\newcommand{\Proc}{\mathsf{Proc}}

\newcommand{\GCTT}{\textbf{GCTT}}
\newcommand{\GDTT}{\textbf{GDTT}}
\newcommand{\CTT}{\textbf{CTT}}
\newcommand{\HoTT}{\textbf{HoTT}}
\newcommand{\TCTT}{\textbf{TCTT}}
\newcommand{\refl}{\mathsf{refl}}
\newcommand{\ap}{\mathsf{ap}}
\newcommand{\Pfin}{\mathsf{P_{fin}}}

\renewcommand{\I}{\mathbb{I}}
\newcommand{\Iw}{\I_\omega}
\newcommand{\ID}{\I_{\DCat}}
\newcommand{\F}{\mathbb{F}}
\newcommand{\Fw}{\F_\omega}
\newcommand{\FD}{\F_{\DCat}}

\newcommand{\hcomp}{\mathsf{hcomp}}

\renewcommand{\C}{\widehat{\mathcal{C}}}
\newcommand{\Cw}{\widehat{\mathcal{C} \times \omega}}
\newcommand{\CD}{\widehat{\mathcal{C} \times \DCat}}
\newcommand{\Path}{\mathsf{Path}}
\newcommand{\nl}{\mathsf{nl}}

\newcommand{\assoc}{\mathsf{assoc}}
\newcommand{\com}{\mathsf{com}}
\newcommand{\idem}{\mathsf{idem}}
\newcommand{\trunc}{\mathsf{trunc}}

\newcommand{\latersymb}{\triangleright}

\newcommand{\later}{\latersymb}
\newcommand{\tot}{\mathsf{tot}}
\newcommand{\inl}{\mathsf{inl}}
\newcommand{\inr}{\mathsf{inr}}
\newcommand{\interm}[1]{\mathsf{in}_{#1}}

\newcommand{\fold}{\mathsf{fold}}
\newcommand{\ff}{\mathsf{ff}}
\newcommand{\tr}{\mathsf{tt}}
\newcommand{\proc}[1]{\llbracket #1 \rrbracket}
\newcommand{\sq}{\mathsf{sq}}
\newcommand{\fib}{\mathsf{fib}}
\newcommand{\isEquiv}{\mathsf{isEquiv}}
\newcommand{\Equiv}{\mathsf{Equiv}}
\renewcommand{\iff}{\leftrightarrow}
\newcommand{\imp}{\rightarrow}
\newcommand{\Fin}{\mathsf{Fin}}
\newcommand{\DM}{\mathsf{DM}}
\newcommand{\gStr}{\mathsf{Str^g}}
\newcommand{\gStrk}{\mathsf{Str^g_\kappa}}
\newcommand{\pair}[2]{\left(#1,#2\right)}

\newcommand{\tyeq}{\simeq}
\newcommand{\eqdef}{=}

\newcommand{\laterRel}[1]{{#1}^{\later}}

\newcommand{\GLTS}{\text{GLTS}}

\newcommand{\cubes}{\mathcal{C}}

\newcommand{\CCat}{\mathbb{C}}
\newcommand{\DCat}{\mathbb{D}}
\newcommand{\Psh}[1]{\widehat{#1}}

\newcommand{\T}{\mathbb{T}}
\newcommand\tabs[1]{\lambda (#1 : \T).}

\newcommand\tapp[2][\tickA]{#2\,[#1] }

\newcommand{\tickA}{\alpha}
\newcommand{\tickB}{\beta}
\newcommand\latbind[1]{{\latersymb}\, (#1: \T) .}

\newcommand\subst[2]{(#1/#2)}

\newcommand{\Ty}[2]{#1 \vdash #2}
\newcommand{\Tm}[3]{#1 \vdash #2 : #3}
\newcommand{\cwfsub}[2]{#1 [#2]}
\newcommand{\p}{\mathsf{p}}
\newcommand{\q}{\mathsf{q}}
\newcommand{\compr}[2]{#1.#2}
\newcommand{\cpair}[2]{( #1, #2)}

\newcommand{\vdashw}{\vdash_{\cubes\times\omega}}
\newcommand{\Tyw}[2]{#1 \vdashw #2}
\newcommand{\Tmw}[3]{#1 \vdashw #2 : #3}

\newcommand{\vdashD}{\vdash_{\cubes\times\DCat}}
\newcommand{\TyD}[2]{#1 \vdashD #2}
\newcommand{\TmD}[3]{#1 \vdashD #2 : #3}

\newcommand{\extend}[1]{[#1]}

\newcommand{\searlier}{\blacktriangleleft \!}
\newcommand{\slater}{\blacktriangleright\!}

\newcommand{\adj}[2]{#1 \, \dashv \, #2}

\newcommand{\den}[1]{\llbracket #1 \rrbracket}
\newcommand{\wfctx}[1]{#1 \vdash}

\newcommand{\ctxproj}[2]{p_{#1;#2}}
\newcommand{\yon}{y}

\newcommand{\pfam}[1]{\mathsf{p}}

\newcommand{\Fstar}{F$^\star$}

\newcommand{\meet}{\wedge}

\newcommand{\allact}[2]{\left[#1\right] #2}
\newcommand{\existact}[2]{\left<#1\right> #2}

\newcommand{\HML}{\mathsf{HML}}
\newcommand{\sat}[2]{#1 \models #2}

\renewcommand{\Fin}[1]{\mathsf{Fin}(#1)}
\newcommand{\Label}[1]{\mathsf{Act}(#1)}
\newcommand{\CCS}[1]{\mathsf{CCS}(#1)}
\newcommand{\CCSpure}{\mathsf{CCS}}
\newcommand{\indisj}[2]{\mathsf{in}_{#1} (#2)}
\newcommand{\action}{\mathsf{act}}
\newcommand{\actionL}{\mathsf{act_L}}
\newcommand{\actionR}{\mathsf{act_R}}
\newcommand{\actionnu}{\mathsf{act_\nu}}
\newcommand{\synch}{\mathsf{synch}}

\begin{document}

\title[Bisimulation as path type for guarded recursive types]{Bisimulation as path type for guarded recursive types}         


\author{Rasmus Ejlers M{\o}gelberg}
\affiliation{
  \position{Associate Professor}
  \department{Department of Computer Science}              
  \institution{IT University of Copenhagen}            
  \streetaddress{Rued Langgaards Vej 7}
  \city{Copenhagen}
  \postcode{2300}
  \country{Denmark}                    
}
\email{mogel@itu.dk}          

\author{Niccol{\`o} Veltri}
\affiliation{
  \position{Postdoctoral Researcher}
  \department{Department of Computer Science}              
  \institution{IT University of Copenhagen}            
  \streetaddress{Rued Langgaards Vej 7}
  \city{Copenhagen}
  \postcode{2300}
  \country{Denmark}                    
}
\email{nive@itu.dk}         

\begin{abstract}
In type theory, coinductive types are used to represent processes, 
and are thus crucial for the formal verification of non-terminating 
reactive programs in proof assistants based on type theory, such 
as Coq and Agda. Currently, programming and reasoning about 
coinductive types is difficult for two reasons: The
need for recursive definitions to be productive, and the lack of 
coincidence of the built-in identity types and the important notion 
of bisimilarity. 

Guarded recursion in the sense of Nakano has recently been 
suggested as a possible approach to dealing with the problem of 
productivity, allowing this to be encoded in types. Indeed, 
coinductive types can be encoded using a combination of guarded
recursion and universal quantification over clocks. This paper studies
the notion of bisimilarity for guarded recursive types in \emph{Ticked Cubical
Type Theory},
an extension of Cubical Type Theory with guarded recursion. 
We prove that, for any functor, an abstract, category theoretic 
notion of bisimilarity for the final guarded coalgebra 
is equivalent (in the sense of homotopy type theory)
to path equality (the primitive notion of equality in cubical type 
theory). As a worked example we study a guarded
notion of labelled transition systems, and show that,
as a special case of the general theorem, path equality coincides 
with an adaptation of the usual notion 
of bisimulation for processes. In particular, this implies that guarded recursion
can be used to give simple equational reasoning proofs of 
bisimilarity. This work should be seen as a step towards obtaining 
bisimilarity as path equality for coinductive types using the 
encodings mentioned above.
\end{abstract}



\maketitle

\section{Introduction}

Programming languages with dependent types such as Coq, Agda, \Fstar\ and Idris are 
attracting increasing attention these years. The main use of dependency is for expressing
predicates on types as used in formal verification. Large scale projects in this area include
the CompCert C compiler, fully verified in Coq~\cite{Leroy06}, and the ongoing Everest 
project~\cite{Everest} 
constructing a fully verified HTTPS stack using \Fstar. Even when full formal verification is not
the goal, dependent types can be used for software development~\cite{IdrisBook}, 
pushing the known advantages of types in programming further.

For these applications, coinductive types are of particular importance, because they
describe processes, i.e., non-terminating and reactive programs. However, programming and 
reasoning about coinductive types is difficult in existing systems with dependent types 
for two reasons. The first is the need for totality in type theory, as required for soundness 
of the logical interpretation
of types. To ensure this, recursively defined data of coinductive type must be 
\emph{productive}~\cite{coquand1993infinite} in the sense that all finite unfoldings can be 
computed in finite time. 

Most proof assistants in use today use syntactic checks to ensure productivity of recursive definitions
of coinductive data, but these checks are not modular, and therefore a certain overhead is associated
with convincing the productivity checker. Guarded recursion in the style 
of~\citet{Nakano:Modality} has been suggested as
a solution to this problem, encoding productivity in types. The idea is to use a modal type operator
$\later$ to encode a time delay in types. This allows productive recursive definitions to be encoded 
as maps of type $\later A \to A$, and thus a fixed point operator taking input of that type allows
for productive recursive programming. These fixed points can be used for programming with 
\emph{guarded recursive types}, i.e., recursive types where the recursion variable only occurs
guarded by $\later$ modalities. Consider for example, the type of guarded streams satisfying 
the type equivalence $\gStr \tyeq \Nat\times \later \gStr$, expressing that the head of the 
stream is immediately available, but the tail takes one time step to compute. In this type, the 
stream of $0$s can be defined using the productive recursive definition $\lambda x . \fold\pair 0x
: \later \gStr \to\gStr$. 
Guarded recursive
types can be used to encode coinductive types~\cite{atkey13icfp}, 
if one allows for the delay modality to be 
indexed by a notion of clocks which can be universally quantified. If, e.g., $\kappa$ is a clock variable
and $\gStrk \tyeq \Nat\times \later^\kappa \gStrk$ is a guarded
recursive type, then $\forall\kappa . \gStrk$
is the coinductive type of streams.  

The second problem for coinductive types in dependent type theory is the notion of equality. 
The natural notion of equality for coinductive types is bisimilarity, but this is not known to coincide
with the build in notion of identity types, and indeed, many authors define bisimilarity as a 
separate notion when working with coinductive 
types~\cite{coquand1993infinite,abel2017interactive}. 
For simple programming with coinductive types it should be the case that these notions coincide,
indeed, it should be the case that the identity type is \emph{equivalent} to the 
bisimilarity type. Such a result would be in the spirit of homotopy type theory, providing
an extensionality principle to coinductive types, similar to function extensionality and 
univalence, which can be viewed as an extensionality principle for universes. 

In this paper we take a step towards such an extensionality principle for coinductive types, by proving
a similar statement for guarded recursive types. Rather than working with identity types in the 
traditional form, we work with path types as used in cubical type theory~\cite{CCHM18}, 
a recent type theory 
based on the cubical model of univalence~\cite{BCH14}. Precisely, we present  
\emph{Ticked Cubical Type Theory} (\TCTT), an extension of cubical type
theory with guarded recursion and ticks, a tool deriving from Clocked Type Theory~\cite{clott}
to be used 
for reasoning about guarded recursive types. 
The guarded fixed point operator satisfies the fixed point unfolding equality up to path equality,
and using this, one can encode guarded recursive types up to type equivalence, as fixed points of 
endomaps on the universe. 

We study a notion of guarded coalgebra, i.e., coalgebra for a functor of the form $F(\later(-))$,
and prove that for the final coalgebra 
the path equality
type is equivalent to a general category theoretic notion of coalgebraic bisimulation for 
$F(\later(-))$. This notion of bisimulation is an adaptation to type theory of a notion defined 
by~\citet{HJ98}. As a running example, we study \emph{guarded  
labelled transition systems} (\GLTS s), i.e., coalgebras for the functor $\Pfin(A \times \later (-))$. 
Here $\Pfin$ is the finite powerset functor, which can be defined~\cite{FGGW18} as a higher 
inductive type (HIT)~\cite{hott}, 
i.e., an inductive type with constructors for elements as well as for equalities. We show that
path equality in the final coalgebra for ${\Pfin(A \times \later (-))}$ coincides with 
a natural adaptation of bisimilarity for finitely branching labelled 
transition systems to their guarded variant. This can be proved either directly using guarded
recursion, or as a consequence of the general theorem mentioned above. 
However, there is a small difference between the abstract category theoretic notion of bisimulation
used in the proof of the general theorem and the concrete one for \GLTS s: 
The concrete formulation is a propositionally
truncated~\cite{hott} version of the abstract one. 
The truncated version is more convenient to work with in the special case, due to the 
set-truncation used in the finite powerset functor. On the other hand, using truncation in the 
general case would break equivalence of bisimilarity and path types. 

%
%
%
As a consequence of the coincidence of bisimilarity and path equality, bisimilarity of processes
can be proved using simple equational reasoning and guarded recursion. We give a few 
examples of that. Moreover, we show how to 
represent Milner's Calculus of
Communicating Systems~\cite{milner1980calculus} as well as 
Hennesy-Milner logic in our type theory. 

The use of the \emph{finite} powerset functor is motivated by the desire to extend this work from guarded
recursive types to coinductive types in future work. It is well known from the set theoretic setting, that 
the unrestricted powerset functor does not have a final coalgebra, but restrictions to subsets of
bounded cardinality do \cite{access,adameklevy}. It is therefore to be expected that a similar restriction is needed
in type theory to model processes with non-determinism as a coinductive type. 
We believe that the results presented
here can be proved also for other cardinalities than finite, such as the countable powerset functor. This allows 
more general notions of processes to be modelled. See Section~\ref{sec:conclusion} for a discussion
of this point. 

We prove consistency of \TCTT\ by constructing a denotational model using the category of 
presheaves over $\cubes \times \omega$. The model combines the constructions of the cubical model
(presheaves over the category of cubes $\cubes$) 
and the topos of trees model of guarded recursion (presheaves over 
$\omega$), and builds on a similar construction used to model Guarded Cubical 
Type Theory~\cite{GCTT},
but extends this to ticks using techniques of~\citet{MM18}. 
Higher inductive types, such as the finite powerset type used in the running example, 
can be modelled using the techniques of~\citet{CHM18}, which adapt easily to the 
model used here.

\subsection{Related work}

Guarded recursion in the form used here originates with~\citet{Nakano:Modality}. 
Aside from the applications
to coinduction mentioned above, guarded recursion has also been 
used to model advanced 
features of programming languages~\cite{BMSS12,bizjak2014model}, 
using an abstract form of step-indexing~\cite{Appel:M01}. 
For this reason, type theory with guarded recursion has been proposed as a metalanguage for reasoning 
about programming languages, and indeed guarded recursion is one of the features of 
IRIS~\cite{Iris}, a framework for higher-order separation logic implemented in Coq. 

Previous work on applications of guarded recursion to programming and reasoning with 
coinductive types has mainly studied the examples of streams~\cite{GDTT} and the lifting 
monad~\cite{GDTT:FPC}. 
To our knowledge, bisimulation and examples involving higher inductive types have not 
previously been studied. The type theory closest to \TCTT\ studied previously is Guarded Cubical
Type Theory (\GCTT)~\cite{GCTT} 
which introduced the idea of fixed point unfoldings as paths. The only difference between the 
two languages is that \GCTT\ 
uses delayed substitutions for reasoning about guarded recursive types, and \TCTT\ replaces 
these by ticks. Ticks can be used to encode delayed substitutions, 
but have better operational behaviour~\cite{clott}. \citet{GCTT} also proved that bisimilarity for 
guarded streams coincides with path equality. 

Sized types~\cite{Hughes:sized:types} 
is a different approach to the problem of encoding productivity in types. The idea is to 
associate abstract ordinals to types indicating the number of possible unfoldings of data of the type. 
The coinductive type is the sized type associated with the size $\infty$. 
Unlike guarded recursion, sized types have been implemented in the proof assistant 
Agda~\cite{Agda},
and this has been used, e.g., for object oriented GUI programming~\cite{abel2017interactive}. On the other hand,
the denotational semantics of guarded recursion is better understood than that of sized types,
and this has the benefit of making it easier to combine with cubical type theory, and in particular to
prove soundness of this combination.

\citet{danielsson2017} studies 
bisimulation for coinductive types using sized types in Agda, proving soundness
of up-to-techniques~\cite{milner1983calculi}. 
Danielsson studies a wider range of systems 
and also weak bisimilarity. In all these cases bisimulation is a separate type, and does not
imply equality as in this work. 

\citet{ahrens2015non} 
encode M-types (types of coinductive trees) in homotopy type theory
and prove that for these bisimilarity logically implies equality. Our result is stronger in two
ways: We consider more general coinductive types (although guarded versions of these), 
and we prove equivalence of types. \citet{vezzosi2017streams} proves that bisimilarity 
for streams implies path equality in Cubical Type Theory extended with 
streams. The proof can most likely be generalised to M-types. 

\subsection{Overview}

Section \ref{sec:ctt} gives a brief introduction to Cubical
Type Theory and the notion of higher inductive types. Cubical type theory is extended
to Ticked Cubical Type Theory (\TCTT) in Section~\ref{sec:tctt}. Section~\ref{sec:pfin} 
introduces the 
finite powerset functor, which is used to define the notion of guarded labelled transition
systems studied in Section~\ref{sec:lts}. Among the results proved there are 
the fact that bisimilarity and path equality coincide in the final coalgebra, and examples
of bisimilarity of processes. 
Section~\ref{sec:coalg} 
presents the general theory of guarded coalgebras, proving the general statement
of equivalence of path equality and bisimilarity, and Section~\ref{sec:equiv} 
relates the abstract notion of 
bisimilarity of Section~\ref{sec:coalg} to the concrete one for guarded labelled 
transition systems of 
Section~\ref{sec:lts}. Finally, Section~\ref{sec:sem} constructs the 
denotational semantics of \TCTT, and
Section~\ref{sec:conclusion} concludes and discusses future work.

\section{Cubical type theory}
\label{sec:ctt}

Cubical type theory (\CTT)~\cite{CCHM18} is an extension of Martin-L\"of type theory
with concepts spawning from the cubical interpretation of homotopy type
theory (\HoTT) \cite{BCH14,CCHM18}. \CTT\ is a dependent type theory
with $\Pi$ types, $\Sigma$ types, sum types, natural numbers and a universe
$\U$. We write $A \to B$ and
$A \times B$ for non-dependent $\Pi$ and $\Sigma$ types
respectively. We also write $a = b$ for judgemental equality of terms.

In \CTT, the identity types of Martin-L\"of type theory are replaced by
a notion of  \emph{path}
types. They correspond to an internalization of the homotopical
interpretation of equalities as paths.  Given $x,y : A$, we write
$\Path_A\,x\,y$ for the type of paths between $x$ and $y$ in $A$, i.e., 
functions from an interval $\I$ to $A$ 
with endpoints $x$ and $y$. The interval is not a type, but still 
name assumptions of the form $i : \I$ can appear in contexts, and there is a 
judgement $\Gamma \vdash r : \I$ which means that 
$r$ is formed using the grammar 
\[
r,s ::= 0 \, | \, 1 \, | \, i \, | \, 1 - r \, | \, r \wedge s \, | \,
r \vee s
\]
using the names appearing in $\Gamma$. Two such 
\emph{dimensions} or \emph{names} are considered equal, if they can
be proved equal using the laws of De Morgan algebra. Formally, 
the interval is thus the free De Morgan algebra and it can be thought of
as the real interval $[0,1]$ with $\meet$ and $\vee$ as the minimum
and maximum operations.

A type in a
context $\Gamma$ containing the names $i_1,\dots,i_n :\I$ should be
thought of as a $n$-dimensional cube.  For example, the type
$i : \I, j : \I \vdash A$ corresponds to a square:
\[
\xymatrix@R=1em@C=1em{
A(i/0)(j/1) \ar[rr]^{A(j/1)} \ar[dd]_{A(i/0)} & & A(i/1)(j/1)
\ar[dd]^{A(i/1)} \\
& A  & \\
A(i/0)(j/0) \ar[rr]_{A(j/0)} & & A(i/1)(j/0) }
\]

Here we have written $A(i/0)$ for the result of substituting $0$ for $i$ in 
$A$. 
Similarly to function spaces, path types have abstraction
and application, denoted $\lambda i.\,p$ and $p\,r$ respectively,
satisfying $\beta$ and $\eta$ equality: $(\lambda i. \,p)\,r = p(i/r)$
and $\lambda i. \,p\,i = p$. The typing rules for these two
operations are given as follows:
\[
\infer{\Gamma \vdash \lambda i. \,p : \Path_A\,(p(i/0))\,(p(i/1))}{\Gamma
  \vdash A & \Gamma , i : \I \vdash p : A} 
\qquad
\infer{\Gamma \vdash p \,r : A}{\Gamma \vdash p : \Path_A\,x\,y &
  \Gamma \vdash r : \I} 
\]
satisfying, for $p :
\Path_A\,x\,y$, the judgemental equalities $p\,0 = x$ and $p\,1 =
y$. These equalities state that $p$ is a path with endpoints $x$ and
$y$.  In what follows, we write
$\Path\,x\,y$ instead of
$\Path_A\,x\,y$ if the type $A$ is clear from context.  
Sometimes, especially in equational proofs, we also write
$x \equiv y$ instead of
$\Path\,x\,y$.  

With the rules given so far, it is possible to prove certain
fundamental properties of path equality. The identity path on a term
$x : A$ is given by: $\refl \,x= \lambda i. \,x :
\Path_A \,x\,x$. The inverse of a path $p :
\Path_A \, x\,y$ is given by:
$p^{-1} = \lambda i.\,p\,(1 - i) :
\Path_A\,y\,x$.  One can prove that functions respect path equality:
 Given $f : A \to B$ and a path $p :
\Path_A\,x\,y$ we have
$\ap_f\,p = \lambda i.\, f\,(p\,i) :
\Path_B\,(f \,x)\,(f\,y)$. Moreover, one can prove extensionality
principles which are generally not provable in standard Martin-L\"of
type theory, like function extensionality: Given two functions
$f , g: \Pi x : A. \, B\,x$ and a family of paths
$p : \Pi x :
A.\,\Path_{B\,x}\,(f\,x)\,(g\,x)$ one can define
$\funext\, f\,g\,p = \lambda i.\, \lambda x.\, p\,x\,i:
\Path_{\Pi x : A.\, B\,x}\,f\,g$.  But other properties characterizing path
equality, such as transitivity or substitutivity, are not provable in
the type theory we have described so far. In order to prove these
properties, \CTT\ admits a specific operation called
\emph{composition}. 
 
In order to define this operation, we first introduce the \emph{face
  lattice} $\F$. Formally, $\F$ is the free distributive lattice
generated by the symbols $(i=0)$ and $(i = 1)$ , and the relation
$(i = 0) \wedge (i = 1) = 0_\F$. This means that its elements are
generated by the grammar:
\[
\varphi, \psi ::= 0_\F \, | \, 1_\F \, | \, (i = 0) \, | \, (i = 1) \,
| \, \varphi \wedge \psi \, | \, \varphi \vee \psi
\]
The judgement $\Gamma \vdash \varphi : \F$ states that $\varphi$
contains only names declared in the context $\Gamma$. 
There is an operation $\Gamma, \varphi$ restricting the context
$\Gamma$, if $\Gamma \vdash \varphi : \F$. Types and terms in context
$\Gamma,\varphi$ are called \emph{partial}. Intuitively, a type in a
restricted context should be thought of as a collection of faces of a
higher dimensional cube. For example, the type $i : \I, j : \I,(i = 0)
\vee (j = 1) \vdash A$ corresponds to the left and top faces of a square:
\[
\xymatrix@R=1em@C=1em{
A(i/0)(j/1) \ar[rr]^{A(j/1)} \ar[dd]_{A(i/0)} & & A(i/1)(j/1) \\
& A  & \\
A(i/0)(j/0) & &}
\]
Given a partial term $\Gamma ,\varphi \vdash u : A$, we write $\Gamma
\vdash t : A [\varphi \mapsto u]$ for the conjunction of the two
judgements $\Gamma \vdash t : A$ and $\Gamma,\varphi \vdash t = u :
A$. In this case we say that the term $t$ extends the partial term $u$
on the extent $\varphi$. This notation can be extended to judgements of the form
${t : A[\varphi_1 \mapsto u_1, \dots, \varphi_n \mapsto u_n]}$, abbreviating one 
typing and $n$ equality judgements. 


The composition operation is defined by the following typing rule:
\[
\infer
{\Gamma \vdash \comp^i_A \,[\varphi \mapsto u]\,u_0 :  A(i/1)[\varphi 
  \mapsto u(i/1)]} 
{
\Gamma \vdash \varphi : \F 
&
\Gamma, i : \I \vdash A
&
\Gamma , \varphi, i : \I \vdash u : A
&
\Gamma \vdash u_0 : A(i/0)[\varphi \mapsto u(i/0)] 
}
\]
Intuitively, the terms $u$ and $u_0$ in the hypothesis specify an open
box in the type $A$: $u_0$ corresponds to the base of the box, while
$u$ corresponds to the sides. The $\comp$ operation constructs a lid
for this open box. Composition can be used to construct, e.g., the composition
of paths $p : \Path_A\,x\,y$ and $q :
 \Path_A \,y\,z$:

%
%
\[
p ; q = \lambda i. \,\comp^j_A \,[(i = 0) \mapsto x , (i = 1) \mapsto
q\,j ] \,(p\,i): \Path_A \,x\,z
\]
%

We also refer to~\cite{CCHM18} for the remaining constructions of 
\CTT, which include systems and a glueing construction. The latter is a 
fundamental ingredient in the proof of the univalence axiom and the
construction of composition for the universe. We will not be needing 
gluing and systems directly in this paper, and so omit them from the 
brief overview. 

\subsection{Higher inductive types}
\label{sec:hits}

Recently, \citet{CHM18} have introduced an extension of
\CTT\ with higher inductive types (HITs). HITs are an important concept
in \HoTT\ which generalize the notion of inductive type. A HIT $A$ can
be thought of as an inductive type in which the introduction rules not
only specify the generators of $A$, but can also specify the
generators of the higher equality types of $A$. Using HITs, one can
define topological spaces such as the circle, the torus and
suspensions in \HoTT\ and develop homotopy
theory in a synthetic way \cite[Ch.~8]{hott}. HITs are also used for
implementing free algebras of signatures specified by operations and
equations, for example the free group on a type, or to construct
quotients of types by equivalence relations.

We now recall how to extend \CTT\ with HITs by showing how to add
propositional truncation. We refer to \cite{CHM18} for the definition
of other HITs and the description of a common pattern for introducing
HITs in \CTT.  In Section \ref{sec:pfin}, we will also present the
finite powerset construction as a HIT.

Given a type $A$, we write $\| A \|$ for the \emph{propositional
  truncation} of $A$. The type $\| A \|$ can have at most one
inhabitant up to path equality, informally, an ``uninformative'' proof
of $A$. In other words, the existence of a term $t : \| A \|$ tells us
that the type $A$ is inhabited, but it does not provide any explicit
inhabitant of $A$. Moreover, any other term $u : \| A \|$ is path
equal to $t$. The introduction rules of $\| A \|$ are the following:
\[
\infer{| x | : \| A \|}{x : A}
\qquad
\infer{\sq\,u\,v\,r : \| A \|}{u , v : \| A \| & r : \I}
\]
with the judgemental equalities $\sq\,u\,v\,0 = u$ and $\sq\,u\,v\,1 =
v$. Notice that the higher constructor $\sq$, which when applied to
terms $u$ and $v$ specifies an element in $\Path\,u\,v$, is treated in
\CTT\ as a point constructor which depends on a name $r : \I$.

Remember that in \CTT\ every type is endowed with a composition
operation. Coquand et al. showed that it is possible to define
composition for a HIT by adding a \emph{homogeneous composition}
as an additional constructor for the type. In the case of propositional truncation, the
latter is introduced by the following rule:
\[
\infer{\Gamma \vdash \hcomp^i_{\|A \|} \, [\varphi \mapsto u]\,u_0 : \| A \| [ \varphi
  \mapsto  u(i/1)] }{\Gamma \vdash A & \Gamma \vdash \varphi : \F &
  \Gamma , \varphi , i 
  :  \I \vdash u : \| A \| & \Gamma \vdash u_0 : \| A \| [ \varphi
  \mapsto  u(i/0)]} 
\]
Notice that $\hcomp$ differs from $\comp$, since in the latter the
type $A$ depends also on the name $i : \I$.
We refer to \cite{CHM18} for the details on the derivation of a
composition operation for $\| A \|$ and for other HITs. 

We now describe the elimination principle of propositional
truncation in \CTT. Given $x :  \| A \| \vdash P\,x$,  a family of terms $x :
A \vdash t\,x : P\,| x |$ and a family of paths 
\[u,v : \| A \|, x :
P\,u , y : P\,v, i : \I \vdash sq\,u\,v\,x\,y\,i : P\,(\sq\,u\,v\,i)[(i
= 0) \mapsto x, (i = 1) \mapsto y],\]
we can define $f : \Pi x : \| A
\|. \, P\,x$ by cases:
\begin{align*}
f\,|x| & = t\,x \\
f\,(\sq\,u\,v\,r) & = sq\,u\,v\, (f\,u)\,(f\,v) \, r 
\end{align*}
plus a case for the $\hcomp$ constructor, see~ \cite{CHM18}.

The addition of propositional truncation and other HITs to \CTT\ is
justified by the existence of these types in the cubical set model of
\CTT~ \cite{CHM18}. 

We conclude this section by briefly describing an auxiliary HIT
$\S$. The type $\S$ is generated by two points $\b_0,\b_1 : \S$ and
two paths $\pa_0,\pa_1$ between $\b_0$ and $\b_1$. Similarly to other
HITs, we also have to add an $\hcomp$ constructor. The induction
principle of $\S$ states that to give a map of type
$\Pi s : \S. \,P\,s$ is the same as to give,  for $j = 0,1$, a point $b_j : P\,\b_j$
and a path \[{i : \I \vdash p_j\,i : P \,(\pa_j\,i)[(i = 0) \mapsto b_0, (i = 1)
\mapsto b_1]}.\]
The type $\S$ is used in the HoTT book for specifying
the 0-truncation operation as a HIT \cite[Ch.~6.9]{hott}. In Section
\ref{sec:pfin}, we will employ the type $\S$ in the $\trunc$
constructor of the finite powerset, that will force the finite
powerset type to be a set.

\subsection{Some basic notions and results}

This section recalls some basic notions and results from homotopy type theory.
All results mentioned are proved in~\cite{hott} unless otherwise stated.
%

We say that a type is \emph{contractible} if it has exactly one
inhabitant. We define a type
$\isContr \,A = \Sigma x : A. \, \Pi y :
A.\,\Path\,x\,y$, which is inhabited if and only if $A$ is
contractible. The unit type $1$ is contractible.

We say that a type is a \emph{proposition}, or it is
\emph{($-1$)-truncated}, if there exists a path between any two of its
inhabitants. We define a type
$\isProp \,A = \Pi x\,y : A. \,
\Path\,x\,y$, which is inhabited if and only if $A$ is a
proposition. We also define a type $\Prop$ consisting of all types
which are propositions. Every contractible type is a proposition. The
empty type $0$ is a proposition. The propositional truncation of a
type is a proposition by construction. The dependent function type
$\Pi x : A. \, B \,x$ is a proposition if and only if $B\,x$ is a
proposition for all $x : A$. Given two types $A$ and $B$, we define $A
\vee B = \| A + B \|$. Moreover, given a type $A$ and a type family
$P$ on $A$, we define $\exists
x : A.\,P\,x = \| \Sigma x : A. \, P\,x\|$. When the latter type is
inhabited, we say that there \emph{merely exists} a term $x : A$ such
that $P\,x$ holds.

We say that a type is a \emph{set}, or it is \emph{0-truncated}, if
there exists at most one path between any two of its inhabitants. We
define a type
$\isSet \,A = \Pi x\,y : A. \,\Pi p\,q : \Path\,x\,y.\,
\Path\,p\,q$, which is inhabited if and only if $A$ is a set. We also
define a type $\Set$ consisting of all types which are sets. Every
proposition is a set. The type $\N$ of natural numbers is a
set. Another example of a set is the finite powerset introduced in
Section \ref{sec:pfin}.

The \emph{fiber} of a function $f : A \to B$ over a term $y : B$ is
given by the type $\fib_f\,y = \Sigma x : A.\, \Path \,(f\,x)\,y$.  We
say that a function $f : A \to B$ is an \emph{equivalence} if all the
fibers of $f$ are contractible. We define a type
$\isEquiv\,f = \Pi y : B.\, \isContr\,(\fib_f\,y)$, which is inhabited
if and only if $f$ is an equivalence. We say that two types are
\emph{equivalent} if there exists an equivalence between them. We also
define a type $\Equiv\,A\,B$ consisting of all functions between $A$
and $B$ which are equivalences.  In order to prove that a function $f$
is an equivalence, it is enough to construct a function $g : B \to A$
such that $g \circ f$ is path equal to the identity function on
$A$ and $f \circ g$ is path equal to the identity function
on $B$. 

A characteristic feature of \HoTT\ is the \emph{univalence axiom},
stating that the canonical map from 
$\Path\,A\,B$ to $\Equiv\,A\,B$ is an equivalence, for all types
$A$ and $B$. The univalence axiom is provable in $\CTT$ \cite{CCHM18}. In the rest of the paper, when we refer to the univalence axiom, we mean direction $\Equiv\,A\,B \to \Path\,A\,B$ of the
equivalence. In other words, when we need to prove that two types are path equal, we will invoke the univalence axiom and prove that the two types are equivalent instead. When $A$ and $B$ are propositions, the univalence axiom implies that the type $\Path\,A\,B$ is equivalent to the type $A \iff B$ of logical equivalences between $A$ and $B$, where $A \iff B = (A \to B) \times (B \to A)$.

We conclude this section with a couple of lemmata that we will employ
later on. These are standard results of \HoTT\ that we state without
proofs.

\begin{lemma}
\label{lem:eqid}
Given a type $X$ and elements $x , y , z : X$, if there exists a path
$p : \Path_X \,x \, y$, then the types $\Path_X \,x\, z$ and
$\Path_X \,y\, z$ are equivalent.
\end{lemma}

\begin{lemma}  
\label{lem:sigmaequiv}
 Let $X,X' : \U$ and $Y : X' \to \U$. If there exists an equivalence
 $f : X \to X'$, then the function $h : ( \Sigma x : X.\, Y (f\,x) ) \to
 \Sigma x' : X'.\, Y x'$ given by $h\,(x,y) = (f\,x \, , \, y)$ is
 also an equivalence.
\end{lemma}

\section{Ticked cubical type theory}
\label{sec:tctt}

We now extend \CTT\ with guarded recursion.
The resulting type theory, called
\emph{ticked cubical type theory} (\TCTT), differs from
guarded cubical type theory (\GCTT) \cite{GCTT}
only by featuring ticks. 
Ticks are an invention deriving from Clocked Type
Theory~\cite{clott} and can be used for reasoning about delayed data. In particular, 
ticks can be used to encode the delayed substitutions of \GCTT, which served the same
purpose, but are simpler and can be given confluent, strongly normalising 
reduction semantics satisfying canonicity. Indeed, these results have been proved
for Clocked Type Theory which includes guarded fixed points~\cite{clott}. 


Formally, \TCTT\ extends \CTT\ with tick assumptions $\tickA : \T$ 
in the context, along with abstraction and application to ticks following
these rules.
\begin{gather*}
\infer{\Gamma,\tickA : \T \vdash}{ \Gamma \vdash}
\qquad
\infer{\Gamma \vdash \latbind\tickA A}{\Gamma, \tickA : \T \vdash A}
\qquad
\infer{\Gamma \vdash \tabs \tickA t : \latbind\tickA A}{\Gamma ,
  \tickA : \T \vdash t : A} 
\\
\infer{\Gamma , \tickB  : \T, \Gamma' \vdash \tapp[\tickB]t :
  A\subst{\tickA}{\tickB}}{\Gamma \vdash t : \latbind\tickA A} 
  \qquad
\infer{\Gamma \vdash \latbind\tickA A : \U}{\Gamma, \tickA : \T \vdash A : \U} 
\end{gather*}
together with $\beta$ and $\eta$ rules:
\[
(\tabs\tickA t)[\beta] = t\subst\tickA\tickB
\qquad
\tabs\tickA{\tapp t} = t
\]
The sort $\T$ of ticks enjoys the same status as the interval $\I$. 
In particular it is not a type. The type $\latbind\tickA A$ should be thought of as 
classifying data of type $A$ that is only available one time step from now. 
Similarly, ticks should be thought of as evidence that time has passed. 
For example, in a context  of the form $\Gamma, \tickA : \T, \Gamma'$, the assumptions
in $\Gamma$ are available for one more time step than those of $\Gamma'$. In the 
application rule, the assumption states that $t$ is a promise of data of type $A$ available 
one time-step after the variables in $\Gamma$ have arrived, and thus the tick $\tickB$,
can be used to open $t$ to an element of type $A$. 
We write $\later A$ for $\latbind \tickA A$
where $\tickA$ does not occur free in $A$. 
Note in particular that the rule for tick application prevents terms like
$\lambda x. \tabs \tickA{\tapp{\tapp x}} : \later \later A \to \later A$ 
being well typed. Such a term in combination with the fixed point operator to 
be introduced below would make any type of the form $\later A$ inhabited,
and logically trivialise a large part of the theory.

The abstraction of the tick $\tickA$ in type $\latbind\tickA A$ makes 
the type behave like a dependent function space between ticks and the
type $A$, similarly to the path type. 
This can be used, e.g., to type the terms of a dependent form of the applicative functor
law:
\begin{align}
\next &= \lambda x.\,\tabs\tickA x : A \to \later A \nonumber \\
\odot &= \lambda f. \, \lambda y.\,\tabs\tickA{f[\alpha] \,(y[\alpha])}
: \later (\Pi x : A.\, B\,x) \to \Pi y : \later A.\, \latbind\tickA B(x/y[\alpha]) 
\label{eq:applicative}
\end{align}

As part of their special status, names and faces are independent of time, in 
the sense that they can always be commuted with ticks as expressed in
the invertible rules
\begin{equation}\label{eq:rulescomm}
\infer={\Gamma,i : \I,\tickA : \T \vdash A}{\Gamma,\tickA : \T,i : \I \vdash
  A}  
\qquad
\infer={\Gamma,i : \I,\tickA : \T \vdash t : A}{\Gamma,\tickA : \T,i : \I \vdash
  t : A}  
\qquad
\infer={\Gamma,\varphi,\tickA : \T \vdash A}{\Gamma,\tickA : \T, \varphi
  \vdash A}  
\qquad
\infer={\Gamma,\varphi,\tickA : \T \vdash t : A}{\Gamma,\tickA : \T, \varphi
  \vdash t : A}  
\end{equation}
This effect could have similarly been obtained by not erasing names and faces from 
the assumption of the tick application rule above, in which case the above rules could 
have been derived. One consequence of these rules is an extensionality principle for 
$\later$ stating the equivalence of types
\begin{equation}\label{eq:later:ext}
\Path_{\latbind\tickA A}\,x\,y \tyeq \latbind\tickA(\Path_{A}\,(\tapp x)\,(\tapp y) )
\end{equation}
as witnessed by the terms
\[
\lambda p.\,\lambda \alpha.\,\lambda i.\, (p\,i)[\alpha] :
\Path_{\latbind\tickA A}\,x\,y \to \latbind\tickA(\Path_{A}\,(\tapp x)\,(\tapp y) )
\]
and
\[
\lambda p.\,\lambda i.\, \lambda \alpha.\, p[\alpha]\,i :
\latbind\tickA(\Path_{A}\,(\tapp x)\,(\tapp y) ) \to \Path_{\latbind\tickA A}\,x\,y
\]
In particular, if $x,y : A$ then 
\begin{equation}\label{eq:later:ext:next}
\Path_{\later A}\,(\next\, x)\,(\next\, y) \tyeq \later(\Path_{A}\,x\,y)
\end{equation}

Note that as a consequence of this, the type of propositions is closed under $\later$:
\begin{lemma}
 Let $A : \later \U$. If $\latbind\tickA {\isProp{(\tapp A)}}$ also $\isProp{(\latbind \tickA {\tapp A})}$.
\end{lemma}

\begin{proof}
 Suppose $x,y : \latbind \tickA {\tapp A}$, we must show that $x \equiv y$. By extensionality it 
 suffices to show $\latbind\tickA{(\tapp x \equiv \tapp y)}$. By assumption there is a $p : 
 \latbind\tickA {\isProp{(\tapp A)}}$ and the term $\tabs\tickA{\tapp p\,(\tapp x)\,(\tapp y)}$
 inhabits the desired type. 
\end{proof}

One additional benefit of the ticks is that the composition operator for 
$\later$ can be defined in type theory. 
To see this, assume
\begin{gather*}
\Gamma \vdash \varphi : \F 
\qquad
\Gamma, i : \I \vdash \latbind\tickA A
\\
\Gamma , \varphi, i : \I \vdash u : \latbind\tickA A
\qquad
\Gamma \vdash u_0 : (\latbind\tickA A)(i/0)[\varphi \mapsto u(i/0)]. 
\end{gather*}
Since these imply assumptions for the composition operator on $A$:
\begin{gather*}
\Gamma, \tickA : \T\vdash \varphi : \F 
\qquad
\Gamma, \tickA : \T, i : \I \vdash A
\\
\Gamma , \tickA : \T, \varphi, i : \I \vdash \tapp u : A
\qquad
\Gamma , \tickA : \T \vdash \tapp{u_0} : A(i/0)[\varphi \mapsto \tapp{u}(i/0)]. 
\end{gather*}
we can define
\[
\Gamma \vdash \comp^i_{\latbind\tickA A} \,[\varphi \mapsto u] \,u_0 = \tabs\tickA{\comp^i_{A} \,[\varphi \mapsto \tapp u ]\,(\tapp{u_0})} :
(\latbind\tickA A)(i/1)
\]
Note that this uses the rules in (\ref{eq:rulescomm}) for exchange of
ticks and names and faces. 
It is easy to see that the term $\comp^i_{\latbind\tickA A} \,[\varphi \mapsto
u] \,u_0$ extends $u(i/1)$ on $\varphi$.

\subsection{Fixed points}

\TCTT\ comes with a primitive fixed point operator mapping 
terms $f$ of type $\later A \to A$ to fixed points of the composition $f \circ \next$.
As in \GCTT, the fixed point equality holds only up to path, not judgemental equality. This is
to ensure termination of the reduction semantics, but is not necessary for logical consistency, as 
verified by the model, in which the fixed point equality holds definitionally. The typing rule gives 
the path, with the end points definitionally equal to the two sides of the fixed point equality.
\[
\infer{\Gamma \vdash \dfix^r x.t : \later A}{\Gamma, x : \later A \vdash t : A \quad \Gamma \vdash r : \I}
\qquad
\dfix^1 x.t = \next\, t\subst x{\dfix^0 x.t}
\]
Note that the above \emph{delayed} fixed point is of type $\later A$ rather than $A$. 
This is to maintain canonicity, but one can define a fixed point $\fix^r x.t$ as 
$t \subst x{\dfix^r x.t}$, which gives the derived rule
\[
\infer{\Gamma \vdash \pfix\, x.t : \Path_A (\fix\, x.t)(t\subst x{\next (\fix\, x.t)})}{\Gamma, x : \later A \vdash t : A}
\]
by defining $\pfix\, x.t\eqdef \lambda i . \fix^i x.t$. We write simply $\fix\, x.t$ for $\fix^0\, x.t$

Fixed points are unique in the following sense. Let $\Gamma, x : \later A \vdash t : A$
and consider a term $a : A$ such that
$p : \Path_A\,a\,(t\subst x{\next\,a})$. Then $a$ is path equal to $\fix\, x.t$ as 
can be proved by guarded recursion: If $e : \later \Path_A\,a\,(\fix\, x. t)$ then 
\[
\lambda i. t\subst x{\tabs\tickA e[\tickA]\,i} : \Path_A\,(t \subst x
{\next \, a})\,(t \subst x{\next (\fix\, x. t)})
\]
and thus we can prove $\Path_A\,a\,(\fix\, x. t)$ by composing this with
$p$ and $(\pfix\, x.t)^{-1}$. 

One application of fixed points is to define \emph{guarded recursive types} as
fixed points of endomaps on the universe. 
For example, one can define a type of guarded streams of
natural numbers as
%
%
$\gStr = \fix \,X. \Nat \times \latbind \tickA{X[\tickA]} : \U$. Since path equality
at the universe is equivalence of types, we obtain an equivalence
\[
\gStr \tyeq \Nat\times \later \gStr
\]
witnessed by terms $\fold : \Nat\times \later \gStr \to \gStr$ and 
$\unfold : \gStr \to \Nat\times \later \gStr$. One can then use the fixed 
point operator for programming with streams and, e.g., define a constant
stream of zeros as $\fix\, x. \fold\pair 0x$. We refer to~\cite{GDTT} for more examples
of programming with guarded recursive types. 

Guarded recursive types are simultaneously initial algebras and final 
coalgebras. This is a standard result~\cite{BMSS12}, but formulated here in terms 
of path equality, so we repeat it. Recall that a small functor is a term
$F : \U \to \U$ together with another term (also called $F$) 
of type $\Pi X,Y : U. (X \to Y) \to FX \to FY$ satisfying the functor laws
up to path equality, i.e., with witnesses of types
\begin{gather*}
\Pi X : U . \Path_{FX \to FX} (\lambda x : FX.x) \, (F(\lambda x : X.x)) \\
\Pi X,Y,Z : U, f: X \to Y, g : Y \to Z. \Path_{FX \to FZ} (Fg\circ Ff) \, (F(g\circ f)) \, .
\end{gather*}

For example, $\later : \U \to \U$ is a functor, with functorial action
given by $\later f = \lambda x . \tabs\tickA{f(\tapp x)} : \later A
\to \later B$ for $f : A \to B$.

\begin{proposition} \label{prop:final:coalg}
 Let $F$ be a small functor. The fixed point $\nuFg \eqdef \fix \,
 X. F(\latbind \tickA{X[\tickA]})$ is an initial algebra and a 
 final coalgebra for the functor $F \circ \later$. We spell out the final coalgebra statement: For any
 $A : \U$ and $f : A \to F(\later A)$, there exists a unique (up to path equality) term 
 $h : A \to \nuFg$ such that 
 \[
 \unfold \circ h \equiv F(\later h) \circ f
 \]
\end{proposition}

\begin{proof}
 The equality to be satisfied by $h$ is equivalent to 
  \[
  h \equiv \fold \circ F(\later h) \circ f
 \]
 and this is satisfied by defining $h$ to be $\fix \,k . \fold \circ F(\lambda x. k\odot x) \circ f$
 using the applicative action (\ref{eq:applicative}) to apply $k: \later (A \to \nuFg)$ 
 to $x : \later A$. Uniqueness follows 
 from uniqueness of fixed points. The initial algebra property can be proved similarly. 
\end{proof}

%



\section{The finite powerset functor}
\label{sec:pfin}

In order to define finitely branching labelled transition systems, we
need to represent finite subsets of a given type. There are several
different ways to describe finite sets and finite subsets in type
theory \cite{CS10}. Recently, \citet{FGGW18} have presented an
implementation of the finite powerset functor in \HoTT\ as a higher
inductive type. Given a type $A$, they construct the
type $\Pfin A$ of finite subsets of $A$ as the free join semilattice
over $A$.  We define the type $\Pfin A$ in \CTT\ following the
pattern for HITs given by \citet{CHM18} that we
summarized in Section \ref{sec:hits}. Notice that the same definition
can be read in the extended type theory \TCTT. Given a type $A$, the
type $\Pfin A$ is introduced by the following constructors:
\[
\infer{\varnothing : \Pfin A}{}
\qquad 
\infer{\{a\} : \Pfin A}{a : A}
\qquad 
\infer{x \cup y : \Pfin A}{x, y : \Pfin A}
\]
\[
\infer{\nl\,x\,r: \Pfin A}{x : \Pfin A & r : \I} 
\quad 
\infer{\assoc\,x\,y\,z\,r: \Pfin A}{x,y,z : \Pfin A &r : \I} 
\quad
\infer{\idem \,a\,r: \Pfin A}{a : A &r : \I}  
\quad 
\infer{\com\,x\,y\,r: \Pfin A}{x,y : \Pfin A & r : \I}  
\]
\[
\infer{\trunc\,f \,r\,s: \Pfin A}{\Gamma \vdash  f : \S \to \Pfin A & r ,s: \I}
\]
with the judgemental equalities:
%
\begin{align*}
 \nl\,x\,0 & = \varnothing \cup x 
&
\nl\,x\,1 & = x
\\
\assoc\,x\,y\,z\,0 & = (x \cup y) \cup z
& 
\assoc\,x\,y\,z\,1 & = x \cup (y \cup z)
\\
\idem\,a\,0 & = \{ a \} \cup \{ a \}
&
\idem\,a\,1 & = \{ a \}
\\
\com\,x\,y\,0 & = x \cup y
&
\com\,x\,y\,1 & = y \cup x
\\
\trunc\,f\,0\,j & = f \,(\pa_0 \,j)
&
\trunc\,f\,1\,j & = f \,(\pa_1 \,j)
\\
\trunc\,f\,i\,0 & = f \,(\b_0)
&
\trunc\,f\,i\,1 & = f \,(\b_1)
\end{align*}

As for the higher constructor $\sq$ of propositional truncation given in
Section \ref{sec:hits}, the higher constructors of $\Pfin A$ are
introduced as point constructors
depending on names in $\I$. For example, the constructor $\nl$
states that the empty set $\varnothing$ is a left unit for the union
operation $\cup$. Given $x : \Pfin A$, we have that $\nl\,x$ is a path
between $\varnothing \cup x$ and $x$. 

The constructor $\trunc$ refers to the HIT $\S$ introduced in the end of Section
\ref{sec:hits}, and forces $\Pfin A$ to be a set. To see this, suppose $x,y : \Pfin A$ and 
$p,q : x\equiv y$. To show $p \equiv q$, define 
$f : \S \to \Pfin A$ by recursion on $\S$:
\[
f\,\b_0 = x \qquad 
f\,\b_1 = y \qquad
j : \I \vdash f\,(\pa_0\,j) = p\, j \qquad
j : \I \vdash f\,(\pa_1\,j) = q\, j.
\]
Then $\trunc\,f : p \equiv q$.


Following the pattern for higher inductive types in \CTT, we also
introduce a constructor $\hcomp$ which imposes a homogenous
composition structure on $\Pfin A$ and allows us to define composition
for $\Pfin A$. Assuming $\Gamma \vdash A$, the $\hcomp$ operation is
given as:
\[
\infer{\Gamma \vdash \hcomp^i_{\Pfin A} \, [\varphi \mapsto u]\,u_0 : \Pfin A \, [ \varphi
  \mapsto  u(i/1)] }{\Gamma \vdash A & \Gamma \vdash \varphi : \F & \Gamma , \varphi , i
  :  \I \vdash u : \Pfin A & \Gamma \vdash u_0 : \Pfin A \, [ \varphi
  \mapsto  u(i/0)]} 
\]
Finally, we will assume that the universe is closed under higher inductive type formers 
such as finite powersets in 
the sense that if $A : \U$ also $\Pfin A : \U$. 
Consistency of the extension
of \TCTT\ with the HITs used in this paper is justified by the 
denotational model presented in Section~\ref{sec:sem}. 

The type $\Pfin A$ comes with a rather complex induction principle,
which is similar to the one described by \citet{FGGW18}. We spell it
out for the case in which the 
type we are eliminating into is a proposition. 
Let $Q : \Pfin A \to
\Prop$. Given $e : Q\,\varnothing$, $s : \Pi a : A.\, Q \,\{ a \}$ and $u
: \Pi x\,y : \Pfin A.\,Q\,x \to Q\, y \to Q \,(x \cup y)$, then there
exists a term $g : \Pi x : \Pfin A.\, Q \,x$ such that:
\begin{align*}
g\,\varnothing & = e &
g\,\{ a \} & = s\,a &
g\,(x \cup y) & = u\,x\,y\,(g\,x)\,(g\,y)
 \end{align*}
The assumption that $Q \, x$ is a proposition means that we do
not have to specify where the higher constructors should be mapped.

%

Using this induction principle of $\Pfin A$, one can 
define a membership predicate $\in \, : A \to \Pfin A \to \Prop$:
\[
a \in \varnothing = \bot 
\qquad 
a \in \{ b \} = \| a \equiv b \| 
\qquad
a \in x \cup y = a \in x \vee a \in y.
\] 
The cases for the other constructors are dealt with in a
straightforward way using the univalence axiom.

Given $x,y : \Pfin A$, we write $x \subseteq y$ for $\Pi a : A.\, a
\in x \to a \in y$.
Using the induction principle of the finite powerset construction, it
is also possible to prove an extensionality principle for finite subsets: two
subsets $x,y : \Pfin A$ are path equal if and only if they contain the
same elements, i.e. if the type $\Pi a : A. \, a \in x \iff
a \in y$ is inhabited~\cite{FGGW18}. 

Using the non-dependent version of the induction principle, it is
possible to prove that $\Pfin$ is a functor. Given $f : A \to B$, we
write $\Pfin f : \Pfin A \to \Pfin B$ for the action of $\Pfin$ on the
function $f$. We have:
\[
\Pfin f\,\varnothing = \varnothing
\qquad
\Pfin f\,\{ a \} = \{ f\,a\}
\qquad
\Pfin f\,(x \cup y) = \Pfin f\,x \cup \Pfin f\,y
\]
The cases for the higher constructors are straightforward. For
example, $\Pfin f\,(\nl \,x\,r) = \nl\,(\Pfin f\,x)\,r$.

We conclude this section by mentioning an auxiliary result 
describing the image of  $\Pfin f$. 


\begin{lemma}
\label{lem:powset:img}
  Let $f : A \to B$ and $x : \Pfin A$. If $a \in x$ then also $f(a) \in \Pfin f\,x$. If $b : B$ is in 
  $\Pfin f\,x$, then there merely exists an $a : A$ such that $a\in x$
  and $b \equiv f(a)$.  
\end{lemma}
\begin{proof}
The first statement can be proved by induction on $x$, and we omit the simple 
verification, focusing on the second statement, which is also proved by 
induction on $x$. 

The case of $x = \varnothing$ implies 
$\Pfin f\,x = \varnothing$ and so the assumption $b \in \Pfin f\,x$ implies absurdity.
If $x = \{a\}$, then $\Pfin f \,x = \{f(a)\}$ and so the assumption
implies $\| b \equiv f(a) \|$. 
Since we are proving a proposition, we can apply induction to the latter 
and obtain a proof $p : b \equiv f(a)$. Then $|\pair{a}{\pair{|\refl\, a|}{p}}|$ 
proves the case.

If $x = x_1 \cup x_2$, then $\Pfin f\,x = \Pfin f \,x_1 \cup \Pfin f \,x_2$ and so 
$b \in \Pfin f\,x_1 \vee b \in \Pfin f\,x_2$. By induction we may thus assume that either 
$b \in \Pfin f\,x_1$ or $b \in \Pfin f\,x_2$, the proof of two cases are symmetric. In the first 
case, by induction, there merely exists an $a$ such that $a\in x_1$
  and $b \equiv f(a)$. Since the former of these implies $a \in x$, this implies the proof of the case. 
\end{proof}

%
%

\section{Guarded Labelled transition systems}
\label{sec:lts}

In this section, we show how to represent a guarded version of 
finitely branching labelled
transition systems in \TCTT. From now on, we omit the attribute
``finitely branching'' since this will be the only kind of system we
will consider in this paper. As we will see in Section~\ref{sec:coalg},
bisimulation for final guarded coalgebras coincides with path equality,
and labelled transition systems provide an interesting test case of this.
In particular, we shall see that proving bisimilarity of processes can
be done using simple equational reasoning in combination with guarded
recursion. 
%
%

A \emph{guarded labelled transition system}, or \emph{\GLTS} for short, consists
of a type $X$ of states, a type $A$ of actions (which we will
assume small in the sense that $A: \U$) and a function
$f : X \to \Pfin (A \times \later X)$.  Given a state $x : X$, a later
state $y : \later X$ and an action $a : A$, we write
$x \stackrel{a}{\to}_f y$ for $(a, y) \in f \,x$.

\begin{example}[label=ex:lts]\label{ex}
Consider the \GLTS\ described pictorially as the following labelled directed graph:
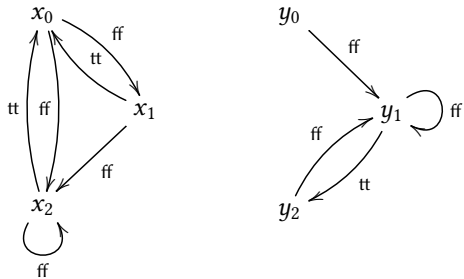
\[
\xymatrix{
x_0 \ar@/^/[dr]^\ff \ar@/^/[dd]_\ff & \\
& x_1 \ar@/^/[ul]_\tr \ar[dl]^\ff \\
x_2 \ar@/^/[uu]^\tr \ar@(dl,dr)_\ff& 
}
\qquad \qquad
\xymatrix{
y_0 \ar[dr]^\ff& \\
& y_1 \ar@/^/[dl]^\tr \ar@(ur,dr)^\ff \\
y_2 \ar@/^/[ur]^\ff & 
}
\]
This can be implemented as the \GLTS\ with $X$ being the
inductive type with six constructors $x_0,x_1,x_2,y_0,y_1,y_2 : X$,
with $A$ being the inductive type with constructors $\ff, \tr : A$, and
\[
\begin{array}{l}
f : X \to \Pfin (A \times \later X) \\
f\,x_0 = \{(\ff,\next\, x_1)\} \cup \{(\ff,\next\,x_2)\} \\
f\,x_1 = \{(\tr,\next\, x_0)\} \cup \{(\ff,\next\,x_2)\} \\
f\,x_2 = \{(\tr,\next\, x_0)\} \cup \{(\ff,\next\,x_2)\} \\
f\,y_0 = \{(\ff,\next\, y_1)\} \\
f\,y_1 = \{(\ff,\next\, y_1)\} \cup \{(\tr,\next\,y_2)\} \\
f\,y_2 = \{(\ff,\next\, y_1)\} 
\end{array}
\]
\end{example}

One is typically interested in \emph{final semantics}, i.e. all
possible runs, or \emph{processes}, of a certain \GLTS. These are obtained by
unwinding a \GLTS\ starting from a particular state. In categorical terms,
the type of processes is given by the final 
coalgebra of the functor $\Pfin\,(A \times \later (-))$, which can be
defined in \TCTT\ as a guarded recursive type using the fixpoint operation:
\[\Proc = \fix\,\,X.\,\Pfin (A \times \latbind\tickA{\tapp X}).\] 
Note that the mapping of $X$ to $\Pfin (A \times \latbind\tickA{\tapp X})$  
has type $\later \U \to \U$ because the universe is closed under $\later$ and
finite powersets as assumed above. As was the case
for the guarded streams example of Section~\ref{sec:tctt}, the fixed point
path induces an equivalence of types witnessed in one direction by 
$\unfold : \Proc \to \Pfin(A \times \later \Proc)$, and thus 
$(\Proc,A,\unfold)$ is a \GLTS .

Given a \GLTS\ $(X,A,f)$, the process associated to a state is
defined as the map $\proc{-} : X \to \Proc$ defined
using the final coalgebra property of Proposition~\ref{prop:final:coalg}. 
%
In the case of Example \ref{ex} it is also possible to define the
processes associated to the \GLTS\ $(X,A,f)$ directly via 
mutual guarded recursion, without using the evaluation function $\proc{-}$.
We can do this since the type $X$ is inductively defined.

\begin{example}[continues=ex:lts]
Let $ps$ and $qs$ be the elements of type $\Proc \times \Proc \times
\Proc$ given as follows:
\[
\begin{array}{ll}
\multicolumn{2}{l}{ps = \fix \, zs.} \\
& \!\! (
\fold \,(\{(\ff, \later \pi_1 \,zs)\} \cup 
\{(\ff, \later \pi_2\, zs)\})  , \\ &
\fold\,(\{(\tr, \later \pi_0 \,zs)\} \cup 
\{(\ff, \later \pi_2\, zs)\})  , \\ &
\fold\,(\{(\tr, \later \pi_0 \,zs)\} \cup 
\{(\ff, \later \pi_2\, zs)\})) \\
\multicolumn{2}{l}{qs = \fix \, zs.} \\
& \!\! (
\fold\,\{(\ff, \later \pi_1 \,zs)\} , \\ &
\fold\,(\{(\ff, \later \pi_1 \,zs)\} \cup 
\{(\tr, \later \pi_2 \,zs)\}) , \\ &
\fold\,\{(\ff, \later \pi_1 \,zs)\})
\end{array}
\]
We define $p_i = \pi_i\,ps$ and $q_i = \pi_i \,qs$, for $i = 0,1,2$. Using
the fixed point operator, one can prove that the
processes $p_i$ and $\proc{x_i}$ are path equal, and
similarly the process $q_i$ and $\proc{y_i}$ are path equal. 
For this, let 
$D = \Path\,p_0\,\proc{x_0} \times
\Path\,p_1\,\proc{x_1} \times
\Path\,p_2\,\proc{x_2}$, and assume $e : \later D$. 
Then
\[
p_0  \equiv \fold\,(\{(\ff, \next \,p_1)\} \cup \{(\ff, \next\, p_2)\}) \equiv \fold\,( \{(\ff,
\next \proc{x_1})\} \cup \{(\ff, \next \proc{x_2})\}) \equiv \proc{x_0}
\]
where the middle path is obtained by first applying $\later \pi_1$ to
$e$ to produce a term in $\later (\Path\,p_1\,\proc{x_1})$, then
using the extensionality principle (\ref{eq:later:ext:next})  for $\later$
to give an element in $\Path\,(\next
\,p_1)\,(\next\,\proc{x_1})$. Analogously, one constructs a path
between $\next\,p_2$
and $\next\,\proc {x_2}$.

Similarly, one shows that $p_1$ and $p_2$ are path equal to
$\proc{x_1}$ and $\proc{x_2}$. Piecing these proofs together we 
get an inhabitant of $t : D$, and thus $\fix\, e. t$ is an element
in $D$ proving the desired path equalities.
\end{example}

\subsection{Bisimulation for \GLTS s}

The natural notion of equality for two processes in a labelled
transition system is bisimilarity.
We now state a notion of bisimilarity
for the notion of \GLTS\ used in this paper. Let 
$(X,A,f)$ be a \GLTS\ and let $R : X \to X \to \U$ be a relation. Then
$R$ is a
\emph{guarded bisimulation} iff the following type is inhabited:
\[
\begin{array}{lc}
\multicolumn{2}{l}{\isLTSBisim_f\,R =\Pi\, x \, y :X.\,R \, x \, y  \to} \\
& (\Pi\,x' : \later X.\,\Pi\,a : A.\, 
 (a,x') \in f\,x 
\to \exists \,y' : \later X.\,(a , y') \in f\,y \times \latbind\tickA
  R\,(\tapp{x'}) \,(\tapp{y'})) \, \\ & \times \\ 
& (\Pi\,y' : \later X.\,\Pi\,a : A.\,  (a,y') \in f\,y 
\to \exists \,x' : \later X.\,(a , x') \in f\,x \times
  \latbind\tickA R\,(\tapp{x'})\,(\tapp{y'})) 
\end{array}
\]

Notice that what we call guarded bisimulation is just the guarded
recursive variant of the usual notion of bisimulation for labelled
transition systems. In words, a relation $R$ is a bisimulation if
whenever two states $x$ and $y$ are related by $R$, two conditions
hold: for all transitions $x \stackrel{a}{\to}_f x'$ there exists a
transition $y \stackrel{a}{\to}_f y'$ such that the later states $x'$
and $y'$ are later related by $R$; for all transitions
$y \stackrel{a}{\to}_f y'$ there exists a transition
$x \stackrel{a}{\to}_f x'$ such that the later states $x'$ and $y'$ are
later related by $R$ . Notice the use of the existential quantifier
$\exists$ instead of $\Sigma$ here. This is necessary for the proofs
of Proposition~\ref{prop:proc:bisim} and Theorem~\ref{thm:coindprinc}
below. 


\begin{proposition} \label{prop:proc:bisim}
 Let $(X,A,f)$ be a \GLTS. Then the relation $R$ defined as 
 $R\,x\,y\eqdef \Path_\Proc \proc x \proc y$ is a bisimulation.
\end{proposition}

\begin{proof}
Suppose $R\,x\,y$ and that $(a,x') \in f(x)$. We show that there merely exists
a $y'$ such that $(a,y') \in f(y)$ and $\latbind\tickA R\,(\tapp{x'})\,(\tapp{y'}))$. The
other direction is proved similarly. By Lemma~\ref{lem:powset:img} 
$(A \times \later \proc -)(a,x')$ is in the finite set
$\Pfin(A \times \later \proc -)(f(x))$. By definition of $\proc -$ as the unique 
map of coalgebras, and the assumption that $\proc x \equiv \proc y$ we get
\begin{align*}
 \Pfin(A \times \later \proc -)(f(x)) & \equiv \unfold(\proc x) \\
 & \equiv \unfold(\proc y) \\
 & \equiv \Pfin(A \times \later \proc -)(f(y)) 
\end{align*}
By Lemma~\ref{lem:powset:img} the property $(A \times \later \proc -)(a,x') \in 
\Pfin(A \times \later \proc -)(f(y))$ implies the mere existence of a $z$
such that $(A \times \later \proc -) z \equiv (A \times \later \proc -)(a,x')$ and 
$z \in f(y)$. Thus $z$ must be 
of the form $(a,y')$,
such that
\begin{align*}
 \tabs \tickA{\proc{\tapp{y'}}} & = \later \proc -(y') \\
  & \equiv \later \proc -(x') \\
  & = \tabs \tickA{\proc{\tapp{x'}}} \, .
\end{align*}
By the extensionality principle (\ref{eq:later:ext}) for $\later$ this implies 
$\latbind \tickA{(\proc{\tapp{x'}}\equiv \proc{\tapp{y'}})}$, which by definition is 
${\latbind\tickA R\,(\tapp{x'})\,(\tapp{y'}))}$ as desired.
\end{proof}

We can also define the greatest guarded bisimulation on a \GLTS\ $(X,A,f)$ by guarded
recursion. 
\[
\begin{array}{lc}
\multicolumn{2}{l}{\sim_f = \fix \,R.\,\lambda x \, y :X.} \\
& (\Pi\,x' : \later X.\,\Pi\,a : A.\, 
 (a,x') \in f\,x 
\to \exists \,y' : \later X.\,(a , y') \in f\,y \times \latbind\tickA
  \tapp R\,(\tapp{x'}) \,(\tapp{y'})) \, \\ & \times \\ 
& (\Pi\,y' : \later X.\,\Pi\,a : A.\,  (a,y') \in f\,y 
\to \exists \,x' : \later X.\,(a , x') \in f\,x \times
  \latbind\tickA{\tapp R\,(\tapp{x'})\,(\tapp{y'}))}
\end{array}
\]
This is indeed the greatest bisimulation in the sense that 
if $R\, x \, y$ holds for some $f$-bisimulation $R$, then also $x\sim_fy$ holds
as can be proved by guarded recursion. 
We call this relation \emph{guarded bisimilarity}. When considering
the \GLTS\ $(\Proc,A,\unfold)$, we write $\sim$ for the guarded
bisimilarity relation $\sim_\unfold$. Notice that the type $x \sim_f
y$ is a proposition, for all $x, y : X$ and $f : X \to \Pfin (A \times
\later X)$.

In \TCTT, the coinduction proof principle can be stated as
follows. Given a guarded bisimulation $R$ for the \GLTS\
$(\Proc, A,\unfold)$, if two processes are related by $R$ then they
are path equal, i.e. for all $p,q : \Proc$ we have an implication
$R\,p\,q \to
\Path_\Proc\,p\,q$. The coinduction proof principle is derivable from
the following theorem, which is a well-known result in the theory of
coalgebras developed in set theory \cite{Rut00}.

\begin{theorem}\label{thm:coindprinc}
  Let $(X,A,f)$ be a \GLTS.  For all $x , y : X$, the types
  $x \sim_f y$ and $\Path_\Proc \proc x \proc y$ are path equal.
\end{theorem}
\begin{proof}
By univalence, it sufficies to show that the two propositions are
logically equivalent. The right-to-left implication follows from
Proposition \ref{prop:proc:bisim} and the fact that $\sim_f$ is the
greatest guarded bisimulation on $(X,A,f)$. For the left-to-right
implication, we proceed by guarded recursion. Suppose ${e : \later
  (\Pi x,y : X. x \sim_f y \to \proc x \equiv \proc y)}$, and let $x,y :
X$ such that $x \sim_f y$ holds. We construct a proof of $\proc x \equiv \proc y$. Note 
that the type $x \sim_f y$ is equivalent to the following type:
\[
\begin{array}{c}
(\Pi\,x' : \later X.\,\Pi\,a : A.\, 
 (a,x') \in f\,x
\to \exists \,y' : \later X.\,(a , y') \in f\,y \times \latbind\tickA
  (\tapp{x'} \sim_f \tapp{y'})) \\ \times \\ 
(\Pi\,y' : \later X.\,\Pi\,a : A.\, 
 (a,y') \in f\,y
\to \exists \,x' : \later X.\,(a , x') \in f\,x \times \latbind\tickA
  (\tapp{x'} \sim_f \tapp{y'}))
\end{array}
\]
By the assumption $e$, there is an implication:
\begin{align*}
\latbind\tickA (\tapp{x'} \sim_f \tapp{y'})
& \to {\latbind\tickA (\proc{\tapp{x'}} \equiv \proc{\tapp{y'}})} 
\end{align*}
which, by extensionality for $\later$, is equal to the type:
\begin{align*}
\latbind\tickA (\tapp{x'} \sim_f \tapp{y'})
& \to \later \proc - (x') \equiv \later \proc - (y')
\end{align*}

So the first line of the unfolding of $x \sim_f y$ implies
\begin{equation}\label{eq:firstline}
(\Pi\,x' : \later X.\,\Pi\,a : A.\, (a,x') \in f\,x \to 
\exists \,y' : \later X.\,(a , y') \in f\,y \times \later \proc - (x')
\equiv \later \proc - (y')).
\end{equation}
We show that this in turn implies
\[
\Pfin(A \times \later \proc -)(f \,x) \subseteq \Pfin(A \times \later
\proc -)(f \,y).
\]
In fact, suppose given an inhabitant $h$ of type (\ref{eq:firstline})
above. Let $a : A$ and $p : \later \Proc$ such that
$(a,p) \in \Pfin(A \times \later \proc -)(f \,x)$. By
Lemma~\ref{lem:powset:img} this implies the mere existence of a $z $
such that $z \in f\,x$ and $(A \times \later \proc -)z \equiv
(a,p)$. Thus $z$ must be 
of the form
$(a,x')$ such that $p \equiv \later \proc - x'$.  By the assumption
$h$, there merely exists $y' : \later X$ such that $(a,y') \in f \,y$
and $\later \proc - x' \equiv \later \proc - y'$. By Lemma
~\ref{lem:powset:img}, this implies
$(A \times \later \proc -)(a,y') \in \Pfin(A \times \later \proc -)(f
\,y)$. Since $(A \times \later \proc -)(a,y') \equiv (a , p)$ we
conclude that $(a,p) \in \Pfin(A \times \later \proc -)(f \,y)$.

Similarly the second line of the unfolding of $x \sim_f y$ implies
$\Pfin(A \times \later \proc -)(f \,y) \subseteq \Pfin(A \times \later
\proc -)(f \,x)$.  Therefore, by extensionality for finite powersets,
$x \sim_f y$ implies
$\Pfin(A \times \later \proc -)(f \,x) \equiv \Pfin(A \times \later
\proc -)(f \,y)$, which in turn implies 
\[
  \proc x \equiv \fold \,(\Pfin(A \times \later \proc -)(f \,x)) \equiv
  \fold \,(\Pfin(A \times \later \proc -)(f \,y)) \equiv \proc y. \qedhere
\]
\end{proof}

\begin{corollary}\label{cor:coindglts}
For all $p , q : \Proc$, the types $p \sim q$ and $p \equiv q$
are path equal.
\end{corollary}
\begin{proof}
This follows from Theorem \ref{thm:coindprinc} and the fact that the
unique map of coalgebras
$\proc - : \Proc \to \Proc$ is path equal to the identity function.
\end{proof}
\begin{example}[continues=ex:lts]
  It is not difficult to see that the processes $p_0$ and $q_0$ (or
  alternatively, $\proc{x_0}$ and $\proc{y_0}$) are bisimilar, and
  therefore equal by the coinduction proof principle. But it is 
  simpler to prove them equal directly by guarded recursion.
%
%
  Note that by Proposition~\ref{prop:proc:bisim} this implies $x_0$ and $y_0$ bisimilar.
  To do this, let 
  $D = \Path\,p_0\,q_0 \times \Path\,p_1\,q_1\times
  \Path\,p_0\,q_2 \times \Path \, p_2 \, q_1$. Assuming $e : \later D$, 
  we must construct four proofs, one for each path type in $D$. We just construct a term
  $e_0 $ of type
  $\Path\,p_0\, q_0$, the other proofs are given in a similar manner. We have the following sequence of equalities:
\begin{align*}
 p_0 & \equiv
\fold\,( \{(\ff,  \next \, p_1) \} \cup  \{(\ff, \next\, p_2) \} )\\
& \equiv  \fold\,(\{(\ff,  \next \, q_1) \} \cup  \{(\ff, \next\, q_1) \} )\\
& \equiv  \fold\, \{(\ff,  \next \, q_1 )\} \\
& \equiv  q_0
\end{align*}
where the second equality follows from the assumption $\later D$ which by 
(\ref{eq:later:ext:next}) implies $\next \, p_1 \equiv  \next \, q_1$ and $\next \, p_2 \equiv  \next \, q_1$.
\end{example}

\subsection{CCS}
\label{sec:CCS}

As an extended example of a \GLTS\ we now show how to represent the syntax of 
Milner's Calculus of Communicating Systems (CCS)~\cite{milner1980calculus}. 
We consider a version with guarded recursion, i.e., processes can be defined recursively with 
the restriction that the recursive variable may only occur under an action. More precisely,
we consider the grammar
\[
P ::= 0 \mid a.P \mid P + P \mid P \| P \mid \nu a. P \mid X \mid \mu X. P
\]
with the restriction that in $\mu X . P$, the variable $X$ may only occur under an action
$a . (-)$. For example, $\mu X . a. X$ is a well formed process, but $\mu X. (P \| X)$
(which would correspond to replication $!P$ as in the $\pi$-calculus~\cite{milner1999pi}) is not. 

For simplicity, we use a De Bruijn representation of processes. Names will be simply numbers
and we will define, for each $n : \Nat$ a type $\CCS{n} : \U$ of closed CCS terms whose freely
occurring names are among the first $n$ numbers. For this we use the inductive family
$\mathsf{Fin} : \Nat \to \U$, where $\Fin n$ contains all natural
numbers strictly smaller than $n$.
The type of actions is then $\Label{n} = \Fin{n} + \Fin{n} + 1$. Following standard conventions,
we write simply $m$ for $\indisj{1}{m}$, $\bar m$ for $\indisj{2}m$ and $\tau$ for $\indisj{3}\star$. 

The inductive family of \emph{closed} CCS terms should satisfy the type equivalence
\begin{equation} \label{eq:CCS:tyeq}
\CCS n \tyeq 1 + \Label n\times \later \CCS n + \CCS n \times \CCS n
+ \CCS n \times \CCS n + \CCS{n+1} 
\end{equation}
stating that a closed CCS term can either be $0$, an action, a binary sum, a parallel composition or a 
name abstraction. The use of $\later$ in the case of actions allows for the definition of guarded recursive
processes, e.g., $\mu X . a. X$ can be represented as $\fix \, x. \indisj 2{a,x}$. In the case of 
name abstraction, the abstracted process can have one more free name than the result. 

To define $\CCSpure$, first consider 
\[
F : \later (\Nat \to \U) \to (\Nat \to U) \to \Nat \to \U 
\]
defined as
\[
F\, X \,Y \, n \eqdef 
1 + \Label n\times \latbind{\tickA}{(\tapp X(n))}  + Y(n) \times Y(n)
+ Y(n) \times Y(n) + Y(n+1) 
\]
and define 
\[
\CCSpure \eqdef \fix\, X . \mu Y. F\,X\,Y
\]
where $\mu$ refers to the inductive family. Inductive types and families are special cases of HITs, which 
the universe is closed under by assumption. 
Equation (\ref{eq:CCS:tyeq}) is then satisfied by 
unfolding the fixed point and inductive family once.

%

By definition of $\fix$, we have
\[
\CCSpure = \mu Y. F\,(\CCSpure')\,Y
\]
where $\CCSpure' = \dfix\, X . \mu Y. F\,X\,Y$. Using the path from
$\dfix\, X . \mu Y. F\,X\,Y$ to $\next(\fix\, X . \mu Y. F\,X\,Y)$ we obtain a type
equivalence
$\latbind \tickA {\tapp {\CCSpure'}}(n) \tyeq \later \CCS n$. 
Let $g : \latbind \tickA {\tapp {\CCSpure'}}(n) \to \later \CCS n$
be the map underlying this equivalence.

The set of CCS actions is recursively defined as a 
GLTS $\action : \CCS{n} \to
\Pfin (\Label{n} \times \later \CCS{n})$
\begin{align*}
\action(0) & = \varnothing \\
\action(a.P) & = \{ (a , g\,P)\} \\
\action(P_1 + P_2) & = \action(P_1) \cup \action(P_2) \\
\action(P_1 \| P_2) & = \actionL(\action(P_1), P_2) \cup
                      \actionR(P_1,\action(P_2)) \cup
                      \synch(\action(P_1),\action(P_2)) \\
\action(\nu a. P) & = \actionnu(\action(P))
\end{align*}
The auxiliary functions 
\begin{align*}
\actionL & : \Pfin(\Label{n} \times \later 
\CCS{n}) \times \CCS n \to \Pfin (\Label{n} \times \later \CCS{n}) \\
\actionR & : \CCS n \times \Pfin
(\Label{n} \times \later 
\CCS{n})\to \Pfin (\Label{n} \times \later \CCS{n}) 
\end{align*}
are given by
\begin{align*}
\actionL(P,u) &\eqdef \Pfin (\Label{n} \times \later (- \| P))(u) & 
\actionR(P,u) &\eqdef \Pfin (\Label{n} \times \later (P \| -))(u). 
\end{align*}
The auxiliary function 
\[
\synch : \Pfin (\Label{n} \times \later
\CCS{n})^2 \to \Pfin (\Label{n} \times \later \CCS{n})\]
is recursively
defined as
\begin{align*}
\synch (\varnothing,v) & = \varnothing \\
\synch (\{ (a , P) \} , \varnothing) & = \varnothing \\
\synch (\{ (m , P) \} , \{ (\bar{m} , Q) \}) & = \{ (\tau , \tabs
                                               \tickA \tapp P \| \tapp Q) \} \\
\synch (\{ (\bar{m} , P) \} , \{ (m , Q) \}) & = \{ (\tau , \tabs 
                                               \tickA \tapp P \| \tapp Q \} \\
\synch (\{ (a , P) \} , \{ (b , Q) \}) & = \varnothing \qquad \text{(if
                                         $a$ and $b$ do not fit the
                                         previous two cases)}\\
\synch (\{ (a , P) \} , v_1 \cup v_2) & = \synch (\{ (a , P) \} , v_1)
                                        \cup \synch (\{ (a , P) \} , v_2)\\ 
\synch (u_1 \cup u_2 , v) & = \synch (u_1,v) \cup \synch(u_2,v)
\end{align*}
In the definition of $\synch$ we omitted the cases for the
higher path constructors, which are straightforward.


Finally, the auxiliary function 
\[
\actionnu : \Pfin
(\Label{n + 1} \times \later \CCS{n + 1}) \to \Pfin (\Label{n} \times \later
\CCS{n})
\]
 is recursively defined as
\begin{align*}
\actionnu \varnothing & = \varnothing & \actionnu (u_1 \cup u_2) & = \actionnu u_1 \cup \actionnu u_2 \\
\actionnu \{ (n , P) \} & = \varnothing &
\actionnu \{ (\bar n  , P) \} & = \varnothing \\
\actionnu \{ (b , P) \} & = \{ (b , \tabs \tickA \, \nu a. \tapp P) \}  & \text{(if $b \neq n, \bar n$)}
\end{align*}
Again we omitted the cases for the higher path constructors.

Instantiating Theorem \ref{thm:coindprinc} with the GLTS $(\CCS n
, \Label n , \action)$, we obtain the following result.
\begin{theorem}
Let $p,q : \CCS n$, then the types $p \sim_\action q$ and $\proc p \equiv \proc
q$ are equivalent.
\end{theorem}

\subsection{Hennessy-Milner Logic}
\label{sec:HML}

Basic properties of \GLTS s can be expressed in Hennessy-Milner 
logic~\cite{hennessy1980observing}. The grammar for propositions
\[
\phi :: = \tr \mid \ff \mid \phi \meet \phi \mid \phi \vee \phi \mid \allact a\phi \mid \existact a\phi 
\]
(where $a$ ranges over the alphabet $A$) can be encoded as an inductive type
\[
\HML \eqdef \mu X. 1 + 1 + X \times X + X \times X + A \times X + A \times X
\]
in the standard way. We shall use Hennessy-Milner logic as notation for this type, e.g.,
writing $\allact a \phi$ for $\interm 5\pair a\phi$ whenever $a : A$ and $\phi : \HML$.  
For any \GLTS\ $(X, A, f)$ we define the satisfiability relation $\sat{}{} : X \to \HML \to \Prop$
by recursion on the second argument:
\begin{align*}
 \sat x\tr & \iff 1 \\ \sat x\ff & \iff 0 \\
  \sat x{\phi \meet \psi} & \iff (\sat x\phi \times \sat x \psi) \\
  \sat x{\phi \vee \psi} & \iff (\sat x\phi \vee \sat x \psi) \\
 \sat x{\allact a\psi} & \iff \Pi x': \later X . (a,x') \in f(x) \imp \latbind\tickA{(\sat{\tapp{x'})}{\phi}} \\
 \sat x{\existact a\psi} & \iff \exists x': \later X . (a,x') \in f(x) \meet \latbind\tickA{(\sat{\tapp{x'})}{\phi}} 
\end{align*}
Since bisimilarity for $\Proc$ coincides with equality, it is obvious
that two bisimilar processes will satisfy the same propositions from
Hennessy-Milner logic. It is a classical
result~\cite{hennessy1980observing}, that Hennessy-Milner logic is
also complete in the sense that any two processes satisfying the same
propositions are also bisimilar. The argument uses classical logic,
and it is unlikely that it can be reproduced in type theory. Still,
one can use propositions to distinguish between processes.

\begin{example}\label{ex:HML}
Suppose $A$ has elements $a,b,c$  such that $\neg (b\equiv c)$ and 
consider the two processes expressed in CCS terms as
\[ 
p = a . (b + c) \qquad q = a.b + a.c
\]
and encoded as elements of $\Proc$ as
\begin{align*}
p & = \{\pair a{\next (\{\pair b{\next \,\varnothing}\} \cup \{ \pair c{\next \,\varnothing}\})}\} \\
q & = (\{\pair a{\next\{\pair b{\next \,\varnothing}\}} \} \cup
    \{\pair a{\next\{\pair c{\next \,\varnothing}\}}\})   
\end{align*}
In the definition above we omitted application
of the function $\fold$ to improve readability.
It is a classical result that
these are not bisimilar, but in the guarded setting, the type $p \equiv q$ is not false, but logically 
equivalent to $\later 0$, which can be proved as follows. Suppose $p \equiv q$, then
since $p$ satisfies $\allact a {\existact b \tr}$, so does $q$. Since 
$\pair a{\next(\{\pair c{\next (\varnothing)}\}} \in \unfold (q)$ this implies
\[
\later(\exists x' : \later X . (b,x') \in \{\pair c{\next (\varnothing)}\})
\]
Reasoning under the $\later$, this implies $\later (\exists x' : \later X . b \equiv c)$ which implies $\later 0$. 

In the opposite direction, we must prove that $\later 0$ implies $p \equiv q$. By the extensionality
principle (\ref{eq:later:ext}) for $\later$, the proposition $\later 0$ implies $\next (x) \equiv \next (y)$ for all $x,y : X$,
and thus $p \equiv \{\pair a{\next (\varnothing)}\} \equiv q$.
\end{example}

\section{Guarded coalgebras}
\label{sec:coalg}

We now move from finitely branching labelled transition systems to
more general kind of systems. These are specified by a guarded
recursive version of coalgebras, that we call guarded coalgebras. We
define bisimulation for these systems and prove a coinduction proof
principle: the greatest bisimulation is equivalent to path equality.

Let $F : \U \to \U$ be a functor. A \emph{guarded coalgebra} for $F$
is a small type $X$ together with a function $f : X \to F (\later X)$.
Analogously, we could say that a guarded coalgebra is a coalgebra for
the composed functor $F \circ \later$.  Proposition~\ref{prop:final:coalg}
states that the fixed point  $\nuFg = \fix \,X. \, F (\latbind\tickA{\tapp X})$
is a final guarded coalgebra for $F$. 
%
%

We now move to the representation of bisimulations. In the literature
there exist several different notions of coalgebraic bisimulation
\cite{Sta11}. Here we consider the variant introduced by 
\citet{HJ98},
which relies on the notion of relation lifting \cite{KV16}.
We adapt this notion to type theory in a way that is not directly 
a generalisation of $\isLTSBisim$ as used in the previous section.
Rather, we will show in Section~\ref{sec:equiv} that 
the latter is a propositionally truncated version of the notion defined here.
The truncated version is convenient for \GLTS s because of the 
set-truncation used in the powerset functor, but truncating the general
notion would falsify the equivalence of bisimilarity of paths as 
stated in Corollary~\ref{cor:coindprinc} below.

\subsection{Relation lifting}


Given a (proof-relevant) relation $R : X \to Y \to \U$, the
\emph{relation lifting} $\rest{F}R : FX \to FY \to \U$ of $F$ on $R$
is defined as
\[
\rest{F}R\,x\,y \, = \, \Sigma t : F(\tot\,R).\, \Path_{FX} \,(F\pi_0\,t)
\, x \times \Path_{FY} \, (F\pi_1\,t) \, y
\]
for $x : FX$ and $y : FY$. 
Here $\tot\ R$ is the graph of $R$, i.e., 
$\tot\,R = \Sigma x : X.\,\Sigma y : Y.\, R \, x \, y$, and
$\pi_0$ and $\pi_1$ refer to the projections out of the dependent product. 



We first show that the relation lifting of $F$ applied to the identity relation
$\Path_X$ is path equal to
$\Path_{FX}$. A proof of this fact in the classical set theoretic setting can be 
found in \cite{Jac16}. Here we adapt the proof to type theory and prove it in 
the general setting of types that are not necessarily sets. 

\begin{proposition}
\label{prop:functorid}
For all $x,y : FX$, the types $\rest{F} (\Path_X) \,x \,y$ and
$\Path_{FX}\,x\,y$ are path equal.
\end{proposition}
\begin{proof}
We have the following sequence of equalities:
\begin{align}
\rest{F} (\Path_X)\,x\,y  & = \Sigma t : F(\tot\,(\Path_X)).\, \Path_{FX}(F\pi_0\,t) \,x \times \Path_{FX}(F\pi_1\,t) \,y 
\nonumber \\
& \equiv \Sigma t : F(\tot\,(\Path_X)).\, \Path_{FX}(F\pi_0\,t) \,x \times \Path_{FX}(F\pi_0\,t) \,y 
\label{eq:pathprop:1}\\
& \equiv \Sigma t : FX.\, \Path_{FX}\, t \,x \times \Path_{FX} \,t \,y \label{eq:pathprop:2}\\
& \equiv \Path_{FX}\,x\,y \label{eq:pathprop:3}
\end{align}

Equality (\ref{eq:pathprop:1}) follows from the fact that the types
  $\Path_{FX}(F\pi_1\,t) \,y$ and $\Path_{FX}(F\pi_0\,t) \,y$ are path
  equal. This in turns
  follows from Lemma \ref{lem:eqid} and the existence of a path:
 \[ 
   \lambda i .\,F (\lambda s. \,\pi_2 \,s \, i) \,t :
  \Path_{FX}\,(F\pi_0 \, t) \,(F\pi_1 \,t). 
\]

Equality (\ref{eq:pathprop:2}) follows from the univalence axiom and Lemma
  \ref{lem:sigmaequiv} instantiated with the function
  $f : F(\tot\,(\Path_X)) \to F X$, $f = F \pi_0$.  The map $f$ is an
  equivalence since $\pi_0 :
  \tot\,(\Path_X) \to X$ is an equivalence and functors preserve equivalences. The map
  $\pi_0$ is an equivalence with inverse
  $h = \lambda x. \, (x , x , \refl\,x) : X \to \tot\,(\Path_X)$.


Equality (\ref{eq:pathprop:3}) follows from the univalence axiom and the
  existence of an equivalence
\[ 
\lambda p.\, (x ,  \refl\,x  ,  p) : \Path_{FX}\,x\,y
\to \Sigma t : FX.\, \Path_{FX}\, t \,x \times \Path_{FX} \,t \,y 
\]
with inverse 
\[ 
\lambda (t , p , q).  p^{-1} ; q :
\Sigma t : FX.\, \Path_{FX}\, t \,x \times \Path_{FX} \,t \,y \to  \Path_{FX}\,x\,y \qedhere
\]
\end{proof}

We next show that the mapping of $R$ to $\rest{(F\circ\later)}(R)$ factors through next,
a property that allows for the notion of guarded bisimilarity to be defined below as a guarded 
fixed point. 
For this, given $R : \later (X \to Y \to \U)$, write
$\laterRel{R} : \later X \to \later Y \to U$ for the relation
$\laterRel{R} x\,y = \latbind\tickA{\tapp R\,(\tapp x)\,(\tapp y)}$. Note that
the extensionality principle (\ref{eq:later:ext}) for the type former $\later$ 
can be expressed by saying that the type
$\laterRel{(\next\,\Path_X)}$ is path equal to $\Path_{\later X}$.

\begin{lemma}\label{lem:rel:lift:contractive}
 For any functor $F$, relation $R : X \to Y \to \U$, and elements $x : F(\later X)$ and $y : F(\later Y)$, 
 the types $\rest{(F\circ\later)}(R)\,x\,y$ and
$\rest F(\laterRel{(\next\,R)})\,x\,y$ are equivalent.
\end{lemma}

\begin{proof}
 First note that 
 \[
 \later (\tot\,R) \tyeq \Sigma u: \later X. \Sigma v: \later Y . \latbind\tickA{R(\tapp u, \tapp v)}
 \]
 by the maps mapping $t : \later (\tot\,R)$ to $(\later (\pi_0)(t), \later (\pi_1)(t), \later (\pi_2)(t))$ and 
 \[(u,v,p) : \Sigma u: \later X. \Sigma v: \later Y . \latbind\tickA{R(\tapp u, \tapp v)}\] to 
 $\tabs\tickA{(\tapp u, \tapp v, \tapp p)}$. Using this and Lemma \ref{lem:sigmaequiv}, we get 
\begin{align*}
 \rest{(F\circ\later)}(R)\,x\,y & =  
 \Sigma t : F(\later (\tot\,R)).\, \Path_{F(\later X)} \,(F(\later(\pi_0))\,t)\, x \times \Path_{F(\later Y)} \, (F(\later(\pi_1))\,t) \, y\\
 & \tyeq 
 \Sigma t : F(\tot\,\laterRel{(\next\,R)}).\, \Path_{F(\later X)} \,(F(\pi_0)\,t)\, x \times \Path_{F(\later Y)} \, (F(\pi_1)\,t) \, y \\
 & = \rest F(\laterRel{(\next\,R)})\,x\,y \qedhere
\end{align*}
\end{proof}

\subsection{Bisimulation for guarded coalgebras}
\label{sec:bisim}

We now implement the notion of coalgebraic bisimulation of 
\citet{HJ98}. 
For a given guarded coalgebra $f : X \to F(\later X)$, we call a
relation $R : X \to X \to \U$ a \emph{guarded bisimulation} iff the
following type is inhabited:  
\[
\isBisim{F}_f\,R = \Pi\, x\, y :X.\,R \, x \, y \to \rest{(F\circ \later)}(R)\,(f\,x)\,(f\,y)
\]

Given a guarded coalgebra $f : X \to
F (\later X)$, \emph{guarded bisimilarity} for $X$ is the 
greatest bisimulation on $X$ defined as follows: 
\[\sim_f \, = \, \fix\,R.\,\lambda\, x\,y. \,\rest{F} \laterRel{R}
(f\,x)\,(f\,y)\]
which by Lemma~\ref{lem:rel:lift:contractive} satisfies
\[
x\sim_fy \tyeq \rest{(F\circ \later)}(\sim_f)(f\, x)(f\, y)
\]
In particular, $\sim_f$ is a guarded bisimulation.

We now show that guarded bisimilarity for the final coalgebra $\nuFg$ is
equivalent to path equality. This follows from the following 
more general proposition. 
%
%
%
%
%
Recall~\cite[Ch.~4.6]{hott} that a function $f : X \to Y$ is an \emph{embedding},
if the map $\lambda p. \, \lambda i. \,f\,(p\,i)  : \Path_X \,x\,y \to \Path_{Y}
\,(f\,x)\,(f\,y)$ is an equivalence for all $x,y : X$. 
If both $X$ and $Y$ are sets, then the type of embeddings between $X$
and $Y$ is equivalent to the type of injections between those two sets.

\begin{proposition}
\label{prop:bisimeq}
If $f : X \to F (\later X)$ is an embedding, then the type families
$\sim_f$ and $\Path_X$ are path equal as elements of $X \to X \to
\U$.
\end{proposition}
\begin{proof}
  The proof proceeds by guarded recursion. Suppose given 
  $p : \later (\Path_{X \to X \to \U} \, (\sim_f) \,
  (\Path_X))$. By function extensionality, it is sufficient to show
  that the types $x \sim_f y$ and $\Path_X\, x\,y$ are path equal for
  all $x,y :X$. We have the following sequence of equalities:
\begin{align}
x \sim_f y & \equiv  \rest{F} \laterRel{(\next \sim_f)} (f\,x)\,(f\,y)  \nonumber \\
& \equiv \rest{F} \laterRel{(\next \, (\Path_X))} (f\,x)\,(f\,y)   \label{eq:IHapp} \\
& \equiv  \rest{F} (\Path_{\later X}) (f\,x)\,(f\,y) \nonumber \\
& \equiv \Path_{F (\later X)} (f\,x)\,(f\,y) \label{eq:prop2}  \\
& \equiv  \Path_X \,x\,y  \nonumber 
 \end{align}
%
%
%
Equality (\ref{eq:IHapp}) holds since the type families $\laterRel{(\next \sim_f)}$ and
$\laterRel{(\next \, (\Path_X))}$ are path equal by the guarded recursive assumption
$p$ and (\ref{eq:later:ext:next}). 
%
Equality (\ref{eq:prop2}) follows from Proposition~\ref{prop:functorid}.
%
%
%
%
\end{proof}

As a corollary, we obtain an extensionality principle for guarded
recursive types. In particular, this implies the coinduction proof
principle for the functor $F$.
\begin{corollary}\label{cor:coindprinc}
For all $x,y : \nuFg$, the types $x\sim_\unfold y$ and $\Path_{\nuFg}
\,x\,y$ are path equal.
\end{corollary}
\begin{proof}
This follows from Proposition \ref{prop:bisimeq} and the fact that
$\unfold$ is an equivalence of types, so in particular it is an
embedding. 
\end{proof}

\section{Equivalence of bisimulations for \GLTS}
\label{sec:equiv}

In this section we show that the concrete notion of guarded
bisimulation $\isLTSBisim$ given in Section \ref{sec:lts} and a
truncated variant of the
notion of coalgebraic guarded bisimulation $\isBisim{F}$ given in
Section \ref{sec:bisim}, for the functor
$F = \Pfin (A \times -)$, are equivalent. 
This is a reformulation to \TCTT\ of a well-known set theoretic result~\cite{Rut00}. 

First, notice that the types $\isLTSBisim_f\,R$ and
$\isBisim{\Pfin(A \times -)}_f\,R$ are generally not equivalent. In
fact, $\isLTSBisim_f\,R$ is always a proposition, but that is not 
always the case for the type ${\isBisim{\Pfin(A \times -)}_f\,R}$. To see this, 
let
%
$X$ be the inductive type with two distinct elements, $x,y : X$, $A$ the unit type 
and $R$ to be the always true
relation. In this case the type 
$\Pfin(A \times \later X)$ is equivalent to
$\Pfin (\later X)$ and we can consider the coalgebra $f$ mapping 
both $x$ and $y$ to $\{\next (x), \next (y)\}$. The type 
$\isBisim{\Pfin(A \times -)}_f\,R$
is then equivalent to 
\[
\Pi\, z\, w :X.\,\Sigma t : \Pfin(\later (X\times X)). \, \Path_{\Pfin(\later X)}(\Pfin(\later \pi_0)(t))(f\,z)\,\times \Path_{\Pfin(\later X)}(\Pfin(\later \pi_1)(t))(f\,w)
\]
and this type can be inhabited by the two distinct elements $c_1, c_2$ defined as  
%
\begin{align*}
c_1\,z\,w & = (t_1,\refl \, \{\next\,x,\next\,y\}, \refl \,
\{\next\,x,\next\,y\}) \\
c_2\,z\,w & = (t_2,\refl \, \{\next\,x,\next\,y\}, \refl \,
\{\next\,x,\next\,y\}) 
\end{align*}
where $t_1 = \{ \next(x,x),\next(y,y)\}$ and $t_2 = \{ \next(x,y),\next(y,x)\}$

Notice though that despite the fact that $\isLTSBisim_f$ and
$\isBisim{\Pfin(A \times -)}_f$ are generally different, the
corresponding bisimilarity relations for
$f = \unfold : \Proc \to \Pfin (A \times \later \Proc)$ are equal by
Corollary \ref{cor:coindglts} and Corollary \ref{cor:coindprinc}.

In this section, we will prove that the type $\isLTSBisim_f\,R$ is
equivalent to a truncated variant of $\isBisim{\Pfin(A \times -)}_f\,R$. 
For a given functor $F$ and a guarded coalgebra $f : X \to F(\later X)$, we call a
relation $R : X \to X \to \U$ a \emph{guarded truncated bisimulation} iff the
following type is inhabited:  
\[
\isTrBisim{F}_f\,R = \Pi\, x\, y :X.\,R \, x \, y \to 
\exists t : F(\later(\tot\,R)).\, \Path \,(F(\later\pi_0)\,t)
\, (f \,x) \times \Path\, (F(\later \pi_1)\,t) \, (f\,y)
\]
This differs from $\isBisim{F}_f$ just by replacing the $\Sigma$ by an existential 
quantifier. 

%
%

We can also introduce a notion of
\emph{guarded truncated bisimilarity} as follows: 
\[\simTr_f \, = \, \fix\,R.\,\lambda\, x\,y. \, 
\exists t : F(\tot\,\laterRel{R}).\, \Path \,(F\pi_0\,t)
\, (f \,x) \times \Path\, (F\pi_1\,t) \, (f\,y)
\]
When $F\,X$ is a set for all types $X$,  the two notions of
guarded bisimilarity are equivalent on the final coalgebra for $F$.

\begin{proposition}
If $F\,X$ is a set for all types $X$, then the types $x \sim_\unfold y$ and $x
\simTr_\unfold y$ are path equal, for all $x,y : \nuFg$.
\end{proposition}

\begin{proof}
By Corollary \ref{cor:coindprinc}, the types $x \sim_\unfold y$ and
$\Path_{\nuFg}\,x\,y$ are path equal. We prove by guarded recursion that
$\simTr_\unfold$ and $\Path_{\nuFg}$ are path equal as terms of $\nuFg
\to \nuFg \to \U$. Suppose given 
  $p : \later (\Path_{\nuFg \to \nuFg \to \U} \, (\simTr_\unfold) \,
  (\Path_{\nuFg}))$. By function extensionality, it is sufficient to show
  that the types $x \simTr_\unfold y$ and $\Path_{\nuFg}\, x\,y$ are path equal for
  all $x,y :X$. We have the following sequence of equalities:
\begin{align}
x \simTr_\unfold y & \equiv  
\| \rest{F} \laterRel{(\next \simTr_\unfold)} (\unfold\,x)\,(\unfold\,y) \| \nonumber \\ 
& \equiv \| \rest{F} \laterRel{(\next \, (\Path_{\nuFg}))}
  (\unfold\,x)\,(\unfold\,y)  \| \label{eq:IHapp:tr} \\ 
& \equiv \| \Path_{\nuFg} \,x\,y \| \label{eq:same} \\
& \equiv  \Path_{\nuFg} \,x\,y \label{eq:setfunct}
 \end{align} 
Here,  (\ref{eq:IHapp:tr}) follows from the induction hypothesis and
(\ref{eq:later:ext:next}), and (\ref{eq:same}) follows
from the fact that the two untruncated types are path equal, which we
derived \emph{en passant} in the proof of Proposition
\ref{prop:bisimeq}. Finally, (\ref{eq:setfunct})
follows from the fact that $\nuFg$ is a set, since
it is equivalent to the set
$F \,(\later \nuFg)$.
\end{proof}

We now give a characterization of the truncated relation lifting of the functor
$\Pfin (A \times -)$, which is the key that allows us to prove that
the concrete notion of bisimulation for labelled transition systems
introduced in Section \ref{sec:lts} is equivalent to the truncated
variant of the general
coalgebraic notion of bisimulation for the functor
$\Pfin (A \times -)$. First, we prove an auxiliary lemma.

\begin{lemma}\label{lem:aux2}
Given a relation $R : X \to Y \to \U$, two finite subsets $u : \Pfin (A \times X)$, $v : \Pfin(A \times Y)$ and a proof $p
: \Pi x : X. \, \Pi a : A.\, (a,x)\in u \to \exists y : Y. \,(a,y) \in v
\times R\,x\,y$,  we have
an inhabitant of the following type:
\[
\exists t : \Pfin (A \times \tot\,R). \,
\Path_{\Pfin (A \times X)} \, (\Pfin (A \times \pi_0)\, t)\, u \times
\Pfin (A \times \pi_1)\, t \subseteq v.
\]

\end{lemma}
\begin{proof}
We construct a term $d\,u\,v\,p$ by induction on $u$.
If $u = \varnothing$, define $d \,u\,v\,p = | (\varnothing
, \refl\, \varnothing, q) |$, where $q$ is the trivial proof of $\varnothing
\subseteq v$.
If $u = \{ (a , x) \}$, define $p' : \exists y : Y. \,(a,y) \in v \times
R\,x\,y$ as $p\,x\,a\,| \refl \, (a , x)|$. 
Since we are proving a proposition, we can use 
the induction principle for propositional truncation
on $p'$, to get $p' = | (y , q , r) |$.  Now define $d\,u\,v\,p
= | (\{ (a,x,y,r)\} , \refl \, \{ (a , x) \} , q' )|$, where $q' : \{ (a,y) \} \subseteq
v$ follows from $q$.

If $u = u_1 \cup u_2$, by induction, we have proofs $d_1$ and
$d_2$ of the statements of the lemma for 
$u_1$ and $u_2$ respectively. Define:
\begin{align*}
d'_1 & : \exists t
  : \Pfin (A \times \tot\,R). \, \Path_{\Pfin (A \times X)} \, (\Pfin
  (A \times \pi_0)\, t)\, u_1 \times \Pfin (A \times \pi_1)\, t
  \subseteq v\\ 
d'_1 & = d_1\,u_1\,v\,(\lambda x\, a \, q.\, p\, x\,a\, | \inl\,q|) \\
d'_2 & : \exists t
  : \Pfin (A \times \tot\,R). \, \Path_{\Pfin (A \times X)} \, (\Pfin
  (A \times \pi_0)\, t)\, u_2 \times \Pfin (A \times \pi_1)\, t
  \subseteq v \\
d'_2 & = d_2\,u_2\,v\,(\lambda x\, a \, q.\, p\, x\,a\, | \inr\,q|)
\end{align*}
Finally, by the induction principle of propositional truncation we can assume 
$d'_1 = |(t_1 , q_1 , s_1)|$ and $d'_2 = |(t_2 , q_2 , s_2)|$ and define
$d\,u\,v\,p = |(t_1 \cup t_2 \, , \, \lambda i.\, q_1 \, i \cup q_2
\, i \, , \,s )|$, where
\[s : \Pfin (A \times \pi_1)\, t_1 \cup \Pfin (A \times \pi_1)\, t_2
\subseteq v\]
follows from $s_1$ and $s_2$.
\end{proof}

\begin{proposition}\label{prop:ltsbisim}
For all $R : X \to Y \to \U$, $u : \Pfin (A \times X)$ and $v : \Pfin (A \times Y)$, 
the type $\| \rest{\Pfin (A \times -)} R \,u 
  \,v \|$ is path equal to the product type:
\[
\begin{array}{c}
(\Pi\,x : X.\,\Pi\,a : A.\,  (a,x) \in u \to \exists \,y : Y.\,(a , y)\in v \times R \,x\,y) \\ \times \\
(\Pi\,y : Y.\,\Pi\,a : A.\,  (a,y) \in v \to \exists \,x : X.\,(a , x)\in u \times R \,x\,y)
\end{array}
\] 
\end{proposition}
\begin{proof}
By univalence, it suffices to show 
that the two propositions are logically equivalent.
For the left-to-right implication, let $e : \| \rest{\Pfin (A \times
  -)} R \,u \,v \|$. We construct a term
\[
d_1 : \Pi\,x : X.\,\Pi\,a :
  A.\, (a,x) \in u \to \exists \,y : Y.\,(a ,  y)\in v \times R \,x\,y
\]
A term $d_2$ of a similar type proving the second component in the product
can be constructed similarly.
Since we are proving a proposition, we can use 
the induction principle of propositional truncation on $p$.
So let $e = |(t,p_1,p_2)|$, with 
\[
t : \Pfin (A \times \tot\, R) \qquad
p_1 : \Path_{\Pfin (A\times X)} \,(\Pfin (A \times \pi_0) \,t)\, u \qquad
p_2 : \Path_{\Pfin (A\times Y)} \,(\Pfin (A \times \pi_1) \,t)\,
v . 
\]
Suppose $x : X$, $a : A$ and $q : (a ,x)
\in u$. Transporting over the path $p_1$, we obtain a proof $q' :
(a,x) \in \Pfin (A \times \pi_0) \,t$. By Lemma
\ref{lem:powset:img} , we know that there merely exists
$y : Y$ and $r : R \,x \,y$ such that $(a,x,y,r) \in t$. In
particular, we have $(a,y) \in \Pfin (A \times \pi_1)
\,t$. Transporting over the path $p_2$, we obtain $(a ,  y)\in v$.

For the right-to-left implication, by Lemma \ref{lem:aux2} we
obtain a term:
\[
s_1 : \exists t : \Pfin (A \times \tot\,R). \,
\Path_{\Pfin (A \times X)} \, (\Pfin (A \times \pi_0)\, t)\, u \times
\Pfin (A \times \pi_1)\, t \subseteq v.
\]
Similarly, by applying a symmetric version of Lemma \ref{lem:aux2} we
obtain another term:
\[
s_2 : \exists t : \Pfin (A \times \tot\,R). \,
\Path_{\Pfin (A \times Y)} \, (\Pfin (A \times \pi_1)\, t)\, v \times
\Pfin (A \times \pi_0)\, t \subseteq u.
\]
Since we are proving a proposition, we can use the induction principle
of propositional truncation on $s_1$ and $s_2$. So let
$s_1 = | (t_1, p_1 , q_1) |$ and $s_2 = | (t_2 ,
p_2 , q_2) |$. We return $ |(t_1 \cup t_2 , p , q) |:  \| \rest{\Pfin(A
  \times -)} R \, u \, v \| $, where $p : \Path_{\Pfin (A \times X)}\,
(\Pfin (A \times \pi_0) t_1 \cup \Pfin (A \times \pi_0) t_2)\, u$
is taken to be the following sequential composition of equalities:
\[
\Pfin (A \times \pi_0) t_1 \cup \Pfin (A \times \pi_0) t_2
\equiv u \cup \Pfin (A \times \pi_0) t_2
\equiv u
\]
The first of these equalities follows from $p_1$ and the second from
$q_2$. The proof $q$ is constructed in a similar way.
\end{proof}

Notice that classically the proof of Proposition \ref{prop:ltsbisim}
is simpler than the constructive one presented here. In the classical
proof, for the right-to-left direction one first constructs a subset
$t : \Pfin(A \times \tot\,R)$ for which $(a,x,y,r) \in t$ if and only
if $(a,x) \in u$, $(a,y) \in v$ and $r : R\,x\,y$. Then one proves,
using the given hypothesis, that $\Pfin (A \times \pi_0)\,t$ is path
equal to $u$ and $\Pfin (A \times \pi_1)\,t$ is path equal
$v$. Constructively, even assuming that the type $A$ comes with
decidable equality, we cannot proceed in two steps like this, since we
do not have a principle of set comprehension as needed for
constructing the term $t$.  It is possible to construct a term
$t' : \Pfin(A \times X \times Y)$ for which $(a,x,y) \in t'$ if and
only if $(a,x) \in u$ and $(a,y) \in v$, but not to filter out of $t'$
the triples $(a,x,y)$ for which $R\,x\,y$ does not hold, since $R$ in
general is an undecidable relation.

As a direct consequence of Proposition \ref{prop:ltsbisim}, we obtain
the equivalence of two notions of bisimulation for \GLTS s.
\begin{theorem}\label{thm:ltsbisim}
Let $(X,A,f)$ be guarded coalgebra for the functor
$F = \Pfin(A \times -)$.  The types 
$\isTrBisim{F}_f\,R$ and $\isLTSBisim_f\,R$ are path equal.
\end{theorem}
%
%


\section{Denotational semantics for \TCTT}
\label{sec:sem}

In this section, we show consistency of \TCTT\  with higher inductive types 
by constructing a denotational model. The construction is an adaptation 
of the model of \GCTT~\cite{GCTT}, which we extend with ticks using the constructions
of \citet{MM18}. The extension with HITs is done exactly as in~\cite{CHM18}, and 
includes ordinary inductive types, such as 
%
the type $\HML$ used in Section~\ref{sec:HML}. 



\subsection{The category of cubical trees}

We first recall the definition of the category of cubical sets as used 
\cite{CCHM18} to model cubical type theory. 
Let $\{i,j,k,\dots\}$ be a countably
infinite set of names. The \emph{category of cubes} $\cubes$ has
finite sets of names $I,J,K,\dots$ as objects, and as morphisms
$f : J \to I$ functions $f : I \to \DM(J$), where $\DM(J)$ is the free
De Morgan algebra on $J$. These compose by the standard Kleisli
composition. 

We write $\C$ for the category of contravariant presheaves on
$\cubes$, whose objects are called \emph{cubical sets}. In elementary terms,
a cubical set $X$ comprises a family of sets $X(I)$ indexed over 
objects $I$ in $\cubes$ and a family of maps indexed over morphisms
$f : J \to I$ in $\cubes$ mapping an $x \in X(I)$ to $x \cdot f \in X(J)$, 
in a functorial way, i.e., $x\cdot\id = x$ and $(x\cdot f)\cdot g = x\cdot (f\circ g)$.
We often leave $\cdot$ implicit, simply writing $xf$.

The category $\C$ has an interval object $\I$ defined by $\I(I) = \DM(I)$, i.e.,
the Yoneda embedding applied to a singleton set. Recall that 
$\C$, being a presheaf category is a topos~\cite{maclane2012sheaves}, 
and in particular has a subobject classifier $\Omega$. The face
lattice $\F$ is the cubical set defined as the subobject of $\Omega$ 
given as the image of the map
$(-) = 1 : \I \to \Omega$ mapping an element $r : I$ to the truth value
$r = 1$. In the internal language of $\C$, this is
\[
\F \eqdef \{ p : \Omega \mid \exists r : \I . p \iff (r = 1)\}
\]
Faces in cubical type theory are modelled 
as elements of $\F$. This corresponds to the object of \emph{cofibrant}
propositions in~\cite{PittsAM:aximct}.

As proved in~\cite{GCTT} these constructions can be generalised to give 
a model of cubical type theory
in presheaves over any category of the form $\cubes\times \DCat$ as long as 
$\DCat$ has an initial object. We will extend this to a model of \TCTT\ in
the case of $\DCat = \omega$, the partially ordered category of natural 
numbers.   
Like $\C$ the topos $\CD$ has an interval object 
defined as $\ID(I,d) = \DM(I)$ and a face lattice $\FD$ defined similarly
to $\F$ in $\C$. 

Before we describe the model of \TCTT, we first recall the notion of 
category with family, the standard notion of model of dependent
type theory~\cite{dybjer1996}.

\begin{definition} \label{def:cwf}
 A \emph{category with family} (CwF)  comprises
 \begin{itemize}
\item A category $\CCat$. We refer to the objects of $\CCat$ as contexts, and the morphisms
as substitutions
\item For each object $\Gamma$ in $\CCat$ a set of types in context $\Gamma$.
We write $\Ty{\Gamma}{A}$ to mean that $A$ is a type in context $\Gamma$.
\item For each $\Ty\Gamma A$ a set of terms of type $A$. We write $\Tm\Gamma tA$
to mean that $t$ is a term of type $A$. 
\item For each substitution $\gamma : \Delta \to \Gamma$ a pair of reindexing maps associating to each
$\Ty{\Gamma}{A}$ a type $\Ty{\Delta}{\cwfsub A \gamma}$ and to each term 
$\Tm\Gamma tA$ a term $\Tm\Delta{\cwfsub t \gamma}{\cwfsub A \gamma}$ such that
\begin{align*}
\cwfsub{(\cwfsub A \gamma)}\delta & = \cwfsub A{\gamma\delta} & \cwfsub A \id & = A  &
\cwfsub{(\cwfsub t \gamma)}\delta & = \cwfsub t{\gamma\delta} & \cwfsub t \id & = t 
\end{align*}
\item A comprehension map, associating to each $\Ty{\Gamma}{A}$ a context
$\compr\Gamma{A}$, a substitution $\p_A : \compr\Gamma A \to \Gamma$, and a term 
$\Tm{\compr\Gamma A}{\q_A}{\cwfsub A{\p_A}}$
such that for every $\gamma : \Delta \to \Gamma$ and $\Tm\Delta t{\cwfsub{A}{\gamma}}$
there exists a unique $\cpair\gamma t : \Delta \to \compr\Gamma A$ such that
\[
\p_A \circ \cpair\gamma t = \gamma \qquad \cwfsub{\q_A}{\cpair\gamma t} = t
\] 
\end{itemize}
\end{definition}

Categories with family model basic dependent types, and can be extended to models of 
$\Pi$ and $\Sigma$ types etc.~\cite{dybjer1996}.

Any presheaf category $\Psh{\CCat}$ defines a category with family, in 
which the category of contexts is $\Psh{\CCat}$, a type
$\Ty\Gamma A$ is a family of sets $A(C, \gamma)$ indexed
over $C \in \CCat$ and $\gamma\in \Gamma(C)$, together with
maps mapping $x\in A(C, \gamma)$ to $xf \in A(D, \gamma f)$ for $f : D \to C$,
in a functorial way. A term $\Tm\Gamma tA$ is a family of elements 
$t(C,\gamma)\in A(C,\gamma)$ such that $t(C,\gamma) f = t(D, \gamma f)$ for all $f$. We shall write
$\TyD\Gamma A$ and $\TmD\Gamma tA$ for the judgements of types and
terms in the CwF associated with $\CD$.

\subsection{A model of \TCTT}

Following \cite{GCTT}, we construct the model using the internal language
of the CwF $\CD$. First we define the notion of a composition structure.

\begin{definition}
 Let $\TyD\Gamma A$ be given. A composition structure on $A$ is an operation
 taking as input a path $ \TmD{\ID}{\gamma}{\Gamma}$, a face $\TmD{1}{\phi}{\FD}$,
 and terms $\TmD{1}{u}{\extend\phi \to \Pi(i : \ID)\cwfsub A{\gamma(i)}}$
 and $\TmD{1}{u_0}{\cwfsub{A}{\gamma(0)}}$ such that $\phi \imp u\,\star\, 0 = u_0$
 and producing a term $\TmD{1}{c_A\, \gamma\, \phi\, u\, u_0}{\cwfsub A{\gamma(1)}}$
such that $\phi \imp u \, \star \, 1 = c_A\, \gamma\, \phi\, u\, u_0$.
\end{definition}
Note that this makes sense, because the CwF is democratic: contexts $\Gamma$
correspond to types in context $1$ (the terminal object). The definition uses the 
subsingleton $\extend\phi$ associated to an element $\phi : \FD$ defined in the internal logic of $\CD$ as $\{\star \mid \phi\}$. 

\begin{theorem}[\cite{GCTT}]
Let $\DCat$ be a small category with an initial object. 
The CwF defined as follows is a model of cubical type theory.
\begin{itemize}
\item The category of contexts is $\CD$
\item A type in context $\Gamma$ is a pair of a type $\TyD\Gamma A$ and a 
composition structure $c_A$ on $A$.
\item Terms of type $(\TyD\Gamma A, c_A)$ are terms $\TmD\Gamma tA$.
\end{itemize}
\end{theorem}

We now specialise to the case of $\DCat = \omega$ and 
show how to extend this model with ticks to a model of \TCTT, using
the techniques developed in~\cite{MM18} to model the ticks of Clocked Type 
Theory~\cite{clott}. First note that there is an adjunction 
$\adj\searlier\slater$ of endofunctors on $\Cw$ defined as follows
\begin{align*}
 (\slater X)(I,0) & = \{\star\} \\
 (\slater X)(I, n+1) & = X(I,n) \\
 (\searlier X)(I,n) & = X(I,n+1)
\end{align*}
The right adjoint $\slater$ is the obvious adaptation of the interpretation of $\later$ 
in the topos of trees model of guarded recursion~\cite{BMSS12}, and the 
left adjoint $\searlier$ is discussed in the same paper.

Extension of contexts with ticks is modelled using the left adjoint
\[
\den{\wfctx{\Gamma, \tickA : \T}} =\, \searlier\den{\wfctx{\Gamma}}
\]
and tick weakening is modelled by the natural transformation 
${\searlier} \to \id$ induced by the maps $(I, n) \to (I,n+1)$ in 
$\cubes \times \omega$. Using this, one can define a context projection
$\ctxproj{\Gamma}{\Gamma'} : \den{\wfctx{\Gamma,\Gamma'}} \to 
\den{\wfctx{\Gamma}}$ by induction on $\Gamma'$. 

\begin{lemma}
 The tick exchange rules are sound, i.e.,
\begin{align*}
 \den{\Gamma, \tickA : \T, i : \I} & = \den{\Gamma, i : \I, \tickA : \T} & 
 \den{\Gamma, \extend\phi, \tickA : \T} & = \den{\Gamma, \tickA : \T, \extend\phi}
\end{align*}
\end{lemma}

\begin{proof}
 The first of these equalities follows directly from the fact that $\searlier \Iw = \Iw$:
 \[
 \den{\Gamma, \tickA : \T, i : \I} = \searlier \den\Gamma \times \Iw = 
 \searlier \den\Gamma \times \searlier\Iw = \searlier (\den\Gamma \times \Iw) = 
 \den{\Gamma, i : \I, \tickA : \T}
 \]
 For the second of these, note first that $\Gamma, \tickA : \T \vdash \phi : \F$ 
 if and only if $\Gamma \vdash \phi : \F$. The interpretation associates to
 each $(I,n)$ and $\gamma \in \den\Gamma (I, n+1)$, 
 an element $\den{\Gamma, \tickA : \T \vdash \phi : \F}(I,n, \gamma)$
 in $\Fw(I,n)$, which is a subset of the subobject classifier at $(I,n)$. 
 An easy induction shows
 that $\den{\Gamma, \tickA : \T \vdash \phi : \F}(I,n, \gamma)$ is true iff 
 $\den{\Gamma\vdash \phi : \F}(I,n+1, \gamma)$ is true, which implies the second equality.
\end{proof}

Given a
type $\Tyw{\searlier \Gamma}{A}$, we define the semantic later modality
$\Tyw{\Gamma}{\slater_\Gamma A}$ as:
\[
\begin{array}{l}
(\slater_{\Gamma} A)(I,0,\gamma) = \{ \star \} \\
(\slater_{\Gamma} A)(I, n + 1,\gamma) = A(I,n,\gamma)
\end{array}
\]

\begin{lemma}
 If $\Tyw{\searlier \Gamma}{A}$ has a composition structure, so does 
 $\Tyw{\Gamma}{\slater_\Gamma A}$.
\end{lemma}

\begin{proof}
  Given input data $\gamma, \phi, u, u_0$, consider the term 
  $\Tmw{\Iw}{\tilde\gamma}{\searlier\Gamma}$ defined as 
  $\tilde\gamma(I, n, r) = \gamma(I, n+1, r)$. Note that  
  \begin{align*}
    \Tyw{\Iw}{\cwfsub A{\tilde\gamma}}(I,n,r) & = A(I, n, \tilde\gamma(r)) \\
    & = (\slater_\Gamma A) (I, n+1, \gamma(r)) \\
    & = \cwfsub{(\slater_{\Iw} A)}{\gamma}(I,n + 1,r)
  \end{align*}
  so we can define $\Tm{1}{\tilde u_0}{\cwfsub A{\tilde\gamma(0)}}$ as 
  $\tilde u_0(I, n) = u_0(I, n+1)$ and likewise 
  \[\Tmw{1}{\tilde u}{\extend\phi \to \Pi(i : \Iw)\cwfsub A{\tilde\gamma(i)}}\]
  as $\tilde u(I,n) = u(I, n+1)$. Finally, to define 
  $\Tmw{1}{c_{\slater_\Gamma A}\, \gamma\, \phi\, u\, u_0}{\cwfsub{(\slater_\Gamma A)}{\gamma(1)}}$
  we define  
  \begin{align*}
   c_{\slater_\Gamma A}\, \gamma\, \phi\, u\, u_0(I,0) & = \star &
   c_{\slater_\Gamma A}\, \gamma\, \phi\, u\, u_0(I,n+1) 
   & = c_{A}\, \tilde\gamma\, \phi\, \tilde u\, \tilde u_0(I,n) . \qedhere
  \end{align*}
\end{proof}

There is a bijective correspondence between terms $\Tmw\Gamma t{\,\slater_{\Gamma} A}$ and
terms $\Tmw{\searlier \Gamma}uA$ given by 
\begin{align*}
 t(I, 0, \gamma) & = \star & t(I, n+1, \gamma) & = u(I,n, \gamma) & u(I,n,\gamma) & = t(I,n+1, \gamma)
\end{align*}
Writing $\overline{(-)}$ for both directions of this bijection, we define 
\begin{align*}
\den{\Gamma \vdash \latbind\tickA A} & = \, \slater_{\den\Gamma}\den{\Gamma, \tickA : \T \vdash A} \\
\den{\Gamma \vdash \tabs\tickA t} &= \overline{\den{\Gamma,\tickA : \T \vdash t}}  \\
\den{\Gamma, \tickA : \T, \Gamma' \vdash \tapp t} &= \overline{\den{\Gamma \vdash t}} 
\circ\ctxproj{\Gamma, \tickA}{\Gamma'}
\end{align*}
This can be verified to satisfy the $\beta$ and $\eta$ equalities~\cite{MM18}. Likewise
fixed points can be modelled in the standard way~\cite{GCTT}. Indeed the fixed point
equality holds definitionally in the model. 

Higher inductive types can be modelled in $\Cw$ essentially in the same way as 
in~\cite{CHM18}, and we refer the reader to loc.cit. for details. Note that 
in~\cite[Section~2.5]{CHM18} higher inductive types defining operations on types,
such as the finite powerset functor, preserve the universe level, so indeed the universe
is closed under such constructions. 

It remains to show that the universe is closed under the $\latbind\tickA(-)$ operation. 
As is standard, the universe is modelled in~\cite{GCTT} under the assumption
of the existence of a Grothendieck universe in the ambient set theory, by modelling
$\den\U(I,n)$ as the set of small types in context $\yon(I,n)$. Since the operation 
$\slater_\Gamma(-)$ preserves smallness, the universe is closed under 
$\latbind\tickA(-)$.

\section{Conclusion and future work}
\label{sec:conclusion}

We have shown that in the type theory \TCTT\ combining cubical type theory 
with guarded recursion, the notions of bisimilarity and path equality
coincide for guarded coalgebras, in the sense of equivalence of types. 
As a consequence of this, representing processes as guarded labelled transition
systems allows for proofs of bisimilarity to be done 
in a simple way using guarded recursion. 

As stated in the introduction, the use of the finite powerset functor is motivated by the
desire to use this work as a stepping stone towards similar results for coinductive types.
For this, the use of finiteness, as opposed to some other cardinality restriction, appears 
to be non-essential. \citet{veltriphdthesis}
has given a description of the countable powerset functor as a HIT, and we believe that
the results presented here can be extended to this functor as well. This would allow, e.g.,
to extend our presentation of CCS with replication $! P$. 


In future work we plan to extend \TCTT\ with clocks and universal quantification
over these. This should allow for coinductive types to be encoded using 
guarded recursive types~\cite{atkey13icfp}, and for guarded recursion
to be used for coinductive reasoning~\cite{GDTT}. In particular, the reasoning
of Example~\ref{ex:HML}, which shows that the two processes $p$ and $q$ satisfying
$p \equiv q \iff \later 0$ can be used to show that their corresponding elements
in the coinductive type of labelled transition systems are genuinely not equal. 

We expect the general proof of coincidence of bisimilarity and path equality
to lift easily to the coinductive case, but encoding the coinductive type of 
processes as a final coalgebra for the functor $\Pfin(A\times -)$ is a challenge.
This requires $\Pfin$ to commute with
universal quantification over clocks~\cite{atkey13icfp,Mogelberg14}, a result
which holds in the model, but which we are currently unable to express in syntax. 
Doing this will most likely require new syntactical constructions.
Universal quantification over clocks does commute with ordinary inductive types
and W-types, also in the syntax of \GDTT, which allows nested inductive and 
coinductive types to be defined using
guarded recursion. Unfortunately, the techniques 
used in these results do not appear to be applicable to HITs. 


Future work also includes investigating more advanced proofs of 
bisimulation e.g., using 
up-to-techniques~\cite{milner1983calculi,danielsson2017,pous2012enhancements}
and weak bisimulation, which is generally challenging for guarded 
recursion~\cite{GDTT:FPC}. It would also be desirable to have an implementation
of guarded recursion in a proof assistant, which would allow proofs such as 
those presented in this paper to be formally checked by a computer. There does exist
a prototype implementation of Guarded Cubical Type Theory~\cite{GCTT}, 
but we found this inadequate for larger proofs. 
%

\begin{acks}                            
  This work was supported by a research grant (13156) from VILLUM FONDEN,
  and by DFF-Research Project 1 Grant no. 4002-00442, from The Danish Council 
  for Independent Research for the Natural Sciences (FNU). We thank the anonymous 
  referees for valuable comments.
\end{acks}

\bibliography{paper.bib}


\begin{thebibliography}{44}


\ifx \showCODEN    \undefined \def \showCODEN     #1{\unskip}     \fi
\ifx \showDOI      \undefined \def \showDOI       #1{#1}\fi
\ifx \showISBNx    \undefined \def \showISBNx     #1{\unskip}     \fi
\ifx \showISBNxiii \undefined \def \showISBNxiii  #1{\unskip}     \fi
\ifx \showISSN     \undefined \def \showISSN      #1{\unskip}     \fi
\ifx \showLCCN     \undefined \def \showLCCN      #1{\unskip}     \fi
\ifx \shownote     \undefined \def \shownote      #1{#1}          \fi
\ifx \showarticletitle \undefined \def \showarticletitle #1{#1}   \fi
\ifx \showURL      \undefined \def \showURL       {\relax}        \fi
\providecommand\bibfield[2]{#2}
\providecommand\bibinfo[2]{#2}
\providecommand\natexlab[1]{#1}
\providecommand\showeprint[2][]{arXiv:#2}

\bibitem[\protect\citeauthoryear{Abel, Adelsberger, and Setzer}{Abel
  et~al\mbox{.}}{2017}]%
        {abel2017interactive}
\bibfield{author}{\bibinfo{person}{Andreas Abel}, \bibinfo{person}{Stephan
  Adelsberger}, {and} \bibinfo{person}{Anton Setzer}.}
  \bibinfo{year}{2017}\natexlab{}.
\newblock \showarticletitle{Interactive programming in Agda - Objects and
  graphical user interfaces}.
\newblock \bibinfo{journal}{\emph{J. Funct. Program.}}  \bibinfo{volume}{27}
  (\bibinfo{year}{2017}), \bibinfo{pages}{e8}.
\newblock
\urldef\tempurl%
\url{https://doi.org/10.1017/S0956796816000319}
\showDOI{\tempurl}


\bibitem[\protect\citeauthoryear{Ad{\'{a}}mek, Levy, Milius, Moss, and
  Sousa}{Ad{\'{a}}mek et~al\mbox{.}}{2015}]%
        {adameklevy}
\bibfield{author}{\bibinfo{person}{Jir{\'{\i}} Ad{\'{a}}mek},
  \bibinfo{person}{Paul~Blain Levy}, \bibinfo{person}{Stefan Milius},
  \bibinfo{person}{Lawrence~S. Moss}, {and} \bibinfo{person}{Lurdes Sousa}.}
  \bibinfo{year}{2015}\natexlab{}.
\newblock \showarticletitle{On Final Coalgebras of Power-Set Functors and
  Saturated Trees - To George Janelidze on the Occasion of His Sixtieth
  Birthday}.
\newblock \bibinfo{journal}{\emph{Applied Categorical Structures}}
  \bibinfo{volume}{23}, \bibinfo{number}{4} (\bibinfo{year}{2015}),
  \bibinfo{pages}{609--641}.
\newblock
\urldef\tempurl%
\url{https://doi.org/10.1007/s10485-014-9372-9}
\showDOI{\tempurl}


\bibitem[\protect\citeauthoryear{Ahrens, Capriotti, and Spadotti}{Ahrens
  et~al\mbox{.}}{2015}]%
        {ahrens2015non}
\bibfield{author}{\bibinfo{person}{Benedikt Ahrens}, \bibinfo{person}{Paolo
  Capriotti}, {and} \bibinfo{person}{R{\'{e}}gis Spadotti}.}
  \bibinfo{year}{2015}\natexlab{}.
\newblock \showarticletitle{Non-Wellfounded Trees in Homotopy Type Theory}. In
  \bibinfo{booktitle}{\emph{13th International Conference on Typed Lambda
  Calculi and Applications, {TLCA} 2015, July 1-3, 2015, Warsaw, Poland}}
  \emph{(\bibinfo{series}{LIPIcs})},
  \bibfield{editor}{\bibinfo{person}{Thorsten Altenkirch}} (Ed.),
  Vol.~\bibinfo{volume}{38}. \bibinfo{publisher}{Schloss Dagstuhl -
  Leibniz-Zentrum fuer Informatik}, \bibinfo{pages}{17--30}.
\newblock
\urldef\tempurl%
\url{https://doi.org/10.4230/LIPIcs.TLCA.2015.17}
\showDOI{\tempurl}


\bibitem[\protect\citeauthoryear{Appel and McAllester}{Appel and
  McAllester}{2001}]%
        {Appel:M01}
\bibfield{author}{\bibinfo{person}{Andrew~W. Appel} {and}
  \bibinfo{person}{David~A. McAllester}.} \bibinfo{year}{2001}\natexlab{}.
\newblock \showarticletitle{An indexed model of recursive types for
  foundational proof-carrying code}.
\newblock \bibinfo{journal}{\emph{{ACM} Trans. Program. Lang. Syst.}}
  \bibinfo{volume}{23}, \bibinfo{number}{5} (\bibinfo{year}{2001}),
  \bibinfo{pages}{657--683}.
\newblock
\urldef\tempurl%
\url{https://doi.org/10.1145/504709.504712}
\showDOI{\tempurl}


\bibitem[\protect\citeauthoryear{Atkey and McBride}{Atkey and McBride}{2013}]%
        {atkey13icfp}
\bibfield{author}{\bibinfo{person}{Robert Atkey} {and} \bibinfo{person}{Conor
  McBride}.} \bibinfo{year}{2013}\natexlab{}.
\newblock \showarticletitle{Productive coprogramming with guarded recursion}.
  In \bibinfo{booktitle}{\emph{{ACM} {SIGPLAN} International Conference on
  Functional Programming, ICFP'13, Boston, MA, {USA} - September 25 - 27,
  2013}}, \bibfield{editor}{\bibinfo{person}{Greg Morrisett} {and}
  \bibinfo{person}{Tarmo Uustalu}} (Eds.). \bibinfo{publisher}{{ACM}},
  \bibinfo{pages}{197--208}.
\newblock
\urldef\tempurl%
\url{https://doi.org/10.1145/2500365.2500597}
\showDOI{\tempurl}


\bibitem[\protect\citeauthoryear{Bahr, Grathwohl, and M{\o}gelberg}{Bahr
  et~al\mbox{.}}{2017}]%
        {clott}
\bibfield{author}{\bibinfo{person}{Patrick Bahr}, \bibinfo{person}{Hans~Bugge
  Grathwohl}, {and} \bibinfo{person}{Rasmus~Ejlers M{\o}gelberg}.}
  \bibinfo{year}{2017}\natexlab{}.
\newblock \showarticletitle{The clocks are ticking: No more delays!}. In
  \bibinfo{booktitle}{\emph{32nd Annual {ACM/IEEE} Symposium on Logic in
  Computer Science, {LICS} 2017, Reykjavik, Iceland, June 20-23, 2017}}.
  \bibinfo{publisher}{{IEEE} Computer Society}, \bibinfo{pages}{1--12}.
\newblock
\urldef\tempurl%
\url{https://doi.org/10.1109/LICS.2017.8005097}
\showDOI{\tempurl}


\bibitem[\protect\citeauthoryear{Bezem, Coquand, and Huber}{Bezem
  et~al\mbox{.}}{2013}]%
        {BCH14}
\bibfield{author}{\bibinfo{person}{Marc Bezem}, \bibinfo{person}{Thierry
  Coquand}, {and} \bibinfo{person}{Simon Huber}.}
  \bibinfo{year}{2013}\natexlab{}.
\newblock \showarticletitle{A Model of Type Theory in Cubical Sets}. In
  \bibinfo{booktitle}{\emph{19th International Conference on Types for Proofs
  and Programs, {TYPES} 2013, April 22-26, 2013, Toulouse, France}}
  \emph{(\bibinfo{series}{LIPIcs})}, \bibfield{editor}{\bibinfo{person}{Ralph
  Matthes} {and} \bibinfo{person}{Aleksy Schubert}} (Eds.),
  Vol.~\bibinfo{volume}{26}. \bibinfo{publisher}{Schloss Dagstuhl -
  Leibniz-Zentrum fuer Informatik}, \bibinfo{pages}{107--128}.
\newblock
\urldef\tempurl%
\url{https://doi.org/10.4230/LIPIcs.TYPES.2013.107}
\showDOI{\tempurl}


\bibitem[\protect\citeauthoryear{Birkedal, Bizjak, Clouston, Grathwohl,
  Spitters, and Vezzosi}{Birkedal et~al\mbox{.}}{2016}]%
        {GCTT}
\bibfield{author}{\bibinfo{person}{Lars Birkedal}, \bibinfo{person}{Ales
  Bizjak}, \bibinfo{person}{Ranald Clouston}, \bibinfo{person}{Hans~Bugge
  Grathwohl}, \bibinfo{person}{Bas Spitters}, {and} \bibinfo{person}{Andrea
  Vezzosi}.} \bibinfo{year}{2016}\natexlab{}.
\newblock \showarticletitle{Guarded Cubical Type Theory: Path Equality for
  Guarded Recursion}. In \bibinfo{booktitle}{\emph{25th {EACSL} Annual
  Conference on Computer Science Logic, {CSL} 2016, August 29 - September 1,
  2016, Marseille, France}} \emph{(\bibinfo{series}{LIPIcs})},
  \bibfield{editor}{\bibinfo{person}{Jean{-}Marc Talbot} {and}
  \bibinfo{person}{Laurent Regnier}} (Eds.), Vol.~\bibinfo{volume}{62}.
  \bibinfo{publisher}{Schloss Dagstuhl - Leibniz-Zentrum fuer Informatik},
  \bibinfo{pages}{23:1--23:17}.
\newblock
\urldef\tempurl%
\url{https://doi.org/10.4230/LIPIcs.CSL.2016.23}
\showDOI{\tempurl}


\bibitem[\protect\citeauthoryear{Birkedal, M{\o}gelberg, Schwinghammer, and
  St{\o}vring}{Birkedal et~al\mbox{.}}{2012}]%
        {BMSS12}
\bibfield{author}{\bibinfo{person}{Lars Birkedal},
  \bibinfo{person}{Rasmus~Ejlers M{\o}gelberg}, \bibinfo{person}{Jan
  Schwinghammer}, {and} \bibinfo{person}{Kristian St{\o}vring}.}
  \bibinfo{year}{2012}\natexlab{}.
\newblock \showarticletitle{First steps in synthetic guarded domain theory:
  step-indexing in the topos of trees}.
\newblock \bibinfo{journal}{\emph{Logical Methods in Computer Science}}
  \bibinfo{volume}{8}, \bibinfo{number}{4} (\bibinfo{year}{2012}).
\newblock
\urldef\tempurl%
\url{https://doi.org/10.2168/LMCS-8(4:1)2012}
\showDOI{\tempurl}


\bibitem[\protect\citeauthoryear{Bizjak, Birkedal, and Miculan}{Bizjak
  et~al\mbox{.}}{2014}]%
        {bizjak2014model}
\bibfield{author}{\bibinfo{person}{Ale{\v{s}} Bizjak}, \bibinfo{person}{Lars
  Birkedal}, {and} \bibinfo{person}{Marino Miculan}.}
  \bibinfo{year}{2014}\natexlab{}.
\newblock \showarticletitle{A Model of Countable Nondeterminism in Guarded Type
  Theory}. In \bibinfo{booktitle}{\emph{Rewriting and Typed Lambda Calculi -
  Joint International Conference, {RTA-TLCA} 2014, Held as Part of the Vienna
  Summer of Logic, {VSL} 2014, Vienna, Austria, July 14-17, 2014. Proceedings}}
  \emph{(\bibinfo{series}{Lecture Notes in Computer Science})},
  \bibfield{editor}{\bibinfo{person}{Gilles Dowek}} (Ed.),
  Vol.~\bibinfo{volume}{8560}. \bibinfo{publisher}{Springer},
  \bibinfo{pages}{108--123}.
\newblock
\urldef\tempurl%
\url{https://doi.org/10.1007/978-3-319-08918-8\_8}
\showDOI{\tempurl}


\bibitem[\protect\citeauthoryear{Bizjak, Grathwohl, Clouston, M{\o}gelberg, and
  Birkedal}{Bizjak et~al\mbox{.}}{2016}]%
        {GDTT}
\bibfield{author}{\bibinfo{person}{Ales Bizjak}, \bibinfo{person}{Hans~Bugge
  Grathwohl}, \bibinfo{person}{Ranald Clouston}, \bibinfo{person}{Rasmus~Ejlers
  M{\o}gelberg}, {and} \bibinfo{person}{Lars Birkedal}.}
  \bibinfo{year}{2016}\natexlab{}.
\newblock \showarticletitle{Guarded Dependent Type Theory with Coinductive
  Types}. In \bibinfo{booktitle}{\emph{Foundations of Software Science and
  Computation Structures - 19th International Conference, {FOSSACS} 2016, Held
  as Part of the European Joint Conferences on Theory and Practice of Software,
  {ETAPS} 2016, Eindhoven, The Netherlands, April 2-8, 2016, Proceedings}},
  \bibfield{editor}{\bibinfo{person}{Bart Jacobs} {and}
  \bibinfo{person}{Christof L{\"{o}}ding}} (Eds.).
  \bibinfo{publisher}{Springer}, \bibinfo{pages}{20--35}.
\newblock
\urldef\tempurl%
\url{https://doi.org/10.1007/978-3-662-49630-5\_2}
\showDOI{\tempurl}


\bibitem[\protect\citeauthoryear{Brady}{Brady}{2016}]%
        {IdrisBook}
\bibfield{author}{\bibinfo{person}{Edwin Brady}.}
  \bibinfo{year}{2016}\natexlab{}.
\newblock \bibinfo{booktitle}{\emph{Type-driven Development with Idris}}.
\newblock \bibinfo{publisher}{Manning Publications Company}.
\newblock
\showISBNx{9781617293023}
\showLCCN{2017288477}


\bibitem[\protect\citeauthoryear{Cohen, Coquand, Huber, and M{\"o}rtberg}{Cohen
  et~al\mbox{.}}{2018}]%
        {CCHM18}
\bibfield{author}{\bibinfo{person}{Cyril Cohen}, \bibinfo{person}{Thierry
  Coquand}, \bibinfo{person}{Simon Huber}, {and} \bibinfo{person}{Anders
  M{\"o}rtberg}.} \bibinfo{year}{2018}\natexlab{}.
\newblock \showarticletitle{{Cubical Type Theory: A Constructive Interpretation
  of the Univalence Axiom}}. In \bibinfo{booktitle}{\emph{21st International
  Conference on Types for Proofs and Programs (TYPES 2015)}}
  \emph{(\bibinfo{series}{Leibniz International Proceedings in Informatics
  (LIPIcs)})}, \bibfield{editor}{\bibinfo{person}{Tarmo Uustalu}} (Ed.),
  Vol.~\bibinfo{volume}{69}. \bibinfo{publisher}{Schloss
  Dagstuhl--Leibniz-Zentrum fuer Informatik}, \bibinfo{address}{Dagstuhl,
  Germany}, \bibinfo{pages}{5:1--5:34}.
\newblock
\showISBNx{978-3-95977-030-9}
\showISSN{1868-8969}
\urldef\tempurl%
\url{https://doi.org/10.4230/LIPIcs.TYPES.2015.5}
\showDOI{\tempurl}


\bibitem[\protect\citeauthoryear{Coquand}{Coquand}{1993}]%
        {coquand1993infinite}
\bibfield{author}{\bibinfo{person}{Thierry Coquand}.}
  \bibinfo{year}{1993}\natexlab{}.
\newblock \showarticletitle{Infinite Objects in Type Theory}. In
  \bibinfo{booktitle}{\emph{Types for Proofs and Programs, International
  Workshop TYPES'93, Nijmegen, The Netherlands, May 24-28, 1993, Selected
  Papers}} \emph{(\bibinfo{series}{Lecture Notes in Computer Science})},
  \bibfield{editor}{\bibinfo{person}{Henk Barendregt} {and}
  \bibinfo{person}{Tobias Nipkow}} (Eds.), Vol.~\bibinfo{volume}{806}.
  \bibinfo{publisher}{Springer}, \bibinfo{pages}{62--78}.
\newblock
\urldef\tempurl%
\url{https://doi.org/10.1007/3-540-58085-9\_72}
\showDOI{\tempurl}


\bibitem[\protect\citeauthoryear{Coquand, Huber, and M{\"{o}}rtberg}{Coquand
  et~al\mbox{.}}{2018}]%
        {CHM18}
\bibfield{author}{\bibinfo{person}{Thierry Coquand}, \bibinfo{person}{Simon
  Huber}, {and} \bibinfo{person}{Anders M{\"{o}}rtberg}.}
  \bibinfo{year}{2018}\natexlab{}.
\newblock \showarticletitle{On Higher Inductive Types in Cubical Type Theory}.
  In \bibinfo{booktitle}{\emph{Proceedings of the 33rd Annual {ACM/IEEE}
  Symposium on Logic in Computer Science, {LICS} 2018, Oxford, UK, July 09-12,
  2018}}, \bibfield{editor}{\bibinfo{person}{Anuj Dawar} {and}
  \bibinfo{person}{Erich Gr{\"{a}}del}} (Eds.). \bibinfo{publisher}{{ACM}},
  \bibinfo{pages}{255--264}.
\newblock
\urldef\tempurl%
\url{https://doi.org/10.1145/3209108.3209197}
\showDOI{\tempurl}


\bibitem[\protect\citeauthoryear{Danielsson}{Danielsson}{2018}]%
        {danielsson2017}
\bibfield{author}{\bibinfo{person}{Nils~Anders Danielsson}.}
  \bibinfo{year}{2018}\natexlab{}.
\newblock \showarticletitle{Up-to techniques using sized types}.
\newblock \bibinfo{journal}{\emph{{PACMPL}}} \bibinfo{volume}{2},
  \bibinfo{number}{{POPL}} (\bibinfo{year}{2018}),
  \bibinfo{pages}{43:1--43:28}.
\newblock
\urldef\tempurl%
\url{https://doi.org/10.1145/3158131}
\showDOI{\tempurl}


\bibitem[\protect\citeauthoryear{Dybjer}{Dybjer}{1996}]%
        {dybjer1996}
\bibfield{author}{\bibinfo{person}{Peter Dybjer}.}
  \bibinfo{year}{1996}\natexlab{}.
\newblock \showarticletitle{Internal Type Theory}. In
  \bibinfo{booktitle}{\emph{Types for Proofs and Programs, International
  Workshop TYPES'95, Torino, Italy, June 5-8, 1995, Selected Papers}}
  \emph{(\bibinfo{series}{Lecture Notes in Computer Science})},
  \bibfield{editor}{\bibinfo{person}{Stefano Berardi} {and}
  \bibinfo{person}{Mario Coppo}} (Eds.), Vol.~\bibinfo{volume}{1158}.
  \bibinfo{publisher}{Springer}, \bibinfo{pages}{120--134}.
\newblock
\urldef\tempurl%
\url{https://doi.org/10.1007/3-540-61780-9\_66}
\showDOI{\tempurl}


\bibitem[\protect\citeauthoryear{Frumin, Geuvers, Gondelman, and van~der
  Weide}{Frumin et~al\mbox{.}}{2018}]%
        {FGGW18}
\bibfield{author}{\bibinfo{person}{Dan Frumin}, \bibinfo{person}{Herman
  Geuvers}, \bibinfo{person}{L{\'{e}}on Gondelman}, {and}
  \bibinfo{person}{Niels van~der Weide}.} \bibinfo{year}{2018}\natexlab{}.
\newblock \showarticletitle{Finite sets in homotopy type theory}. In
  \bibinfo{booktitle}{\emph{Proceedings of the 7th {ACM} {SIGPLAN}
  International Conference on Certified Programs and Proofs, {CPP} 2018, Los
  Angeles, CA, USA, January 8-9, 2018}},
  \bibfield{editor}{\bibinfo{person}{June Andronick} {and}
  \bibinfo{person}{Amy~P. Felty}} (Eds.). \bibinfo{publisher}{{ACM}},
  \bibinfo{pages}{201--214}.
\newblock
\urldef\tempurl%
\url{https://doi.org/10.1145/3167085}
\showDOI{\tempurl}


\bibitem[\protect\citeauthoryear{Hennessy and Milner}{Hennessy and
  Milner}{1980}]%
        {hennessy1980observing}
\bibfield{author}{\bibinfo{person}{Matthew Hennessy} {and}
  \bibinfo{person}{Robin Milner}.} \bibinfo{year}{1980}\natexlab{}.
\newblock \showarticletitle{On Observing Nondeterminism and Concurrency}. In
  \bibinfo{booktitle}{\emph{Automata, Languages and Programming, 7th
  Colloquium, Noordweijkerhout, The Netherlands, July 14-18, 1980,
  Proceedings}} \emph{(\bibinfo{series}{Lecture Notes in Computer Science})},
  \bibfield{editor}{\bibinfo{person}{J.~W. de~Bakker} {and}
  \bibinfo{person}{Jan van Leeuwen}} (Eds.), Vol.~\bibinfo{volume}{85}.
  \bibinfo{publisher}{Springer}, \bibinfo{pages}{299--309}.
\newblock
\urldef\tempurl%
\url{https://doi.org/10.1007/3-540-10003-2\_79}
\showDOI{\tempurl}


\bibitem[\protect\citeauthoryear{Hermida and Jacobs}{Hermida and
  Jacobs}{1998}]%
        {HJ98}
\bibfield{author}{\bibinfo{person}{Claudio Hermida} {and} \bibinfo{person}{Bart
  Jacobs}.} \bibinfo{year}{1998}\natexlab{}.
\newblock \showarticletitle{Structural Induction and Coinduction in a
  Fibrational Setting}.
\newblock \bibinfo{journal}{\emph{Inf. Comput.}} \bibinfo{volume}{145},
  \bibinfo{number}{2} (\bibinfo{year}{1998}), \bibinfo{pages}{107--152}.
\newblock
\urldef\tempurl%
\url{https://doi.org/10.1006/inco.1998.2725}
\showDOI{\tempurl}


\bibitem[\protect\citeauthoryear{Hughes, Pareto, and Sabry}{Hughes
  et~al\mbox{.}}{1996}]%
        {Hughes:sized:types}
\bibfield{author}{\bibinfo{person}{John Hughes}, \bibinfo{person}{Lars Pareto},
  {and} \bibinfo{person}{Amr Sabry}.} \bibinfo{year}{1996}\natexlab{}.
\newblock \showarticletitle{Proving the Correctness of Reactive Systems Using
  Sized Types}. In \bibinfo{booktitle}{\emph{Conference Record of POPL'96: The
  23rd {ACM} {SIGPLAN-SIGACT} Symposium on Principles of Programming Languages,
  Papers Presented at the Symposium, St. Petersburg Beach, Florida, USA,
  January 21-24, 1996}}, \bibfield{editor}{\bibinfo{person}{Hans{-}Juergen
  Boehm} {and} \bibinfo{person}{Guy L.~Steele Jr.}} (Eds.).
  \bibinfo{publisher}{{ACM} Press}, \bibinfo{pages}{410--423}.
\newblock
\urldef\tempurl%
\url{https://doi.org/10.1145/237721.240882}
\showDOI{\tempurl}


\bibitem[\protect\citeauthoryear{Jacobs}{Jacobs}{2016}]%
        {Jac16}
\bibfield{author}{\bibinfo{person}{Bart Jacobs}.}
  \bibinfo{year}{2016}\natexlab{}.
\newblock \bibinfo{booktitle}{\emph{Introduction to Coalgebra: Towards
  Mathematics of States and Observation}}. \bibinfo{series}{Cambridge Tracts in
  Theoretical Computer Science}, Vol.~\bibinfo{volume}{59}.
\newblock \bibinfo{publisher}{Cambridge University Press}.
\newblock
\showISBNx{9781316823187}
\urldef\tempurl%
\url{https://doi.org/10.1017/CBO9781316823187}
\showDOI{\tempurl}


\bibitem[\protect\citeauthoryear{Jung, Krebbers, Jourdan, Bizjak, Birkedal, and
  Dreyer}{Jung et~al\mbox{.}}{2018}]%
        {Iris}
\bibfield{author}{\bibinfo{person}{Ralf Jung}, \bibinfo{person}{Robbert
  Krebbers}, \bibinfo{person}{Jacques-Henri Jourdan},
  \bibinfo{person}{Ale{\v{s}} Bizjak}, \bibinfo{person}{Lars Birkedal}, {and}
  \bibinfo{person}{Derek Dreyer}.} \bibinfo{year}{2018}\natexlab{}.
\newblock \bibinfo{title}{Iris from the ground up}.  (\bibinfo{year}{2018}).
\newblock
\urldef\tempurl%
\url{https://people.mpi-sws.org/~dreyer/papers/iris-ground-up/paper.pdf}
\showURL{%
\tempurl}


\bibitem[\protect\citeauthoryear{Kurz and Velebil}{Kurz and Velebil}{2016}]%
        {KV16}
\bibfield{author}{\bibinfo{person}{Alexander Kurz} {and} \bibinfo{person}{Jiri
  Velebil}.} \bibinfo{year}{2016}\natexlab{}.
\newblock \showarticletitle{Relation lifting, a survey}.
\newblock \bibinfo{journal}{\emph{J. Log. Algebr. Meth. Program.}}
  \bibinfo{volume}{85}, \bibinfo{number}{4} (\bibinfo{year}{2016}),
  \bibinfo{pages}{475--499}.
\newblock
\urldef\tempurl%
\url{https://doi.org/10.1016/j.jlamp.2015.08.002}
\showDOI{\tempurl}


\bibitem[\protect\citeauthoryear{Leroy}{Leroy}{2006}]%
        {Leroy06}
\bibfield{author}{\bibinfo{person}{Xavier Leroy}.}
  \bibinfo{year}{2006}\natexlab{}.
\newblock \showarticletitle{Formal certification of a compiler back-end or:
  programming a compiler with a proof assistant}. In
  \bibinfo{booktitle}{\emph{Proceedings of the 33rd {ACM} {SIGPLAN-SIGACT}
  Symposium on Principles of Programming Languages, {POPL} 2006, Charleston,
  South Carolina, USA, January 11-13, 2006}},
  \bibfield{editor}{\bibinfo{person}{J.~Gregory Morrisett} {and}
  \bibinfo{person}{Simon L.~Peyton Jones}} (Eds.). \bibinfo{publisher}{ACM},
  \bibinfo{pages}{42--54}.
\newblock
\urldef\tempurl%
\url{https://doi.org/10.1145/1111037.1111042}
\showDOI{\tempurl}


\bibitem[\protect\citeauthoryear{MacLane and Moerdijk}{MacLane and
  Moerdijk}{2012}]%
        {maclane2012sheaves}
\bibfield{author}{\bibinfo{person}{Saunders MacLane} {and}
  \bibinfo{person}{Ieke Moerdijk}.} \bibinfo{year}{2012}\natexlab{}.
\newblock \bibinfo{booktitle}{\emph{Sheaves in geometry and logic: A first
  introduction to topos theory}}.
\newblock \bibinfo{publisher}{Springer Science \& Business Media}.
\newblock


\bibitem[\protect\citeauthoryear{Mannaa and M{\o}gelberg}{Mannaa and
  M{\o}gelberg}{2018}]%
        {MM18}
\bibfield{author}{\bibinfo{person}{Bassel Mannaa} {and}
  \bibinfo{person}{Rasmus~Ejlers M{\o}gelberg}.}
  \bibinfo{year}{2018}\natexlab{}.
\newblock \showarticletitle{The Clocks They Are Adjunctions Denotational
  Semantics for Clocked Type Theory}. In \bibinfo{booktitle}{\emph{3rd
  International Conference on Formal Structures for Computation and Deduction,
  {FSCD} 2018, July 9-12, 2018, Oxford, {UK}}}
  \emph{(\bibinfo{series}{LIPIcs})},
  \bibfield{editor}{\bibinfo{person}{H{\'{e}}l{\`{e}}ne Kirchner}} (Ed.),
  Vol.~\bibinfo{volume}{108}. \bibinfo{publisher}{Schloss
  Dagstuhl--Leibniz-Zentrum fuer Informatik}, \bibinfo{pages}{23:1--23:17}.
\newblock
\urldef\tempurl%
\url{https://doi.org/10.4230/LIPIcs.FSCD.2018.23}
\showDOI{\tempurl}


\bibitem[\protect\citeauthoryear{Milner}{Milner}{1980}]%
        {milner1980calculus}
\bibfield{author}{\bibinfo{person}{Robin Milner}.}
  \bibinfo{year}{1980}\natexlab{}.
\newblock \bibinfo{booktitle}{\emph{A Calculus of Communicating Systems}}.
  \bibinfo{series}{Lecture Notes in Computer Science},
  Vol.~\bibinfo{volume}{92}.
\newblock \bibinfo{publisher}{Springer}.
\newblock
\showISBNx{3-540-10235-3}
\urldef\tempurl%
\url{https://doi.org/10.1007/3-540-10235-3}
\showDOI{\tempurl}


\bibitem[\protect\citeauthoryear{Milner}{Milner}{1983}]%
        {milner1983calculi}
\bibfield{author}{\bibinfo{person}{Robin Milner}.}
  \bibinfo{year}{1983}\natexlab{}.
\newblock \showarticletitle{Calculi for Synchrony and Asynchrony}.
\newblock \bibinfo{journal}{\emph{Theor. Comput. Sci.}}  \bibinfo{volume}{25}
  (\bibinfo{year}{1983}), \bibinfo{pages}{267--310}.
\newblock
\urldef\tempurl%
\url{https://doi.org/10.1016/0304-3975(83)90114-7}
\showDOI{\tempurl}


\bibitem[\protect\citeauthoryear{Milner}{Milner}{1999}]%
        {milner1999pi}
\bibfield{author}{\bibinfo{person}{Robin Milner}.}
  \bibinfo{year}{1999}\natexlab{}.
\newblock \bibinfo{booktitle}{\emph{Communicating and mobile systems: the pi
  calculus}}.
\newblock \bibinfo{publisher}{Cambridge {U}niversity {P}ress}.
\newblock


\bibitem[\protect\citeauthoryear{M{\o}gelberg}{M{\o}gelberg}{2014}]%
        {Mogelberg14}
\bibfield{author}{\bibinfo{person}{R.~E. M{\o}gelberg}.}
  \bibinfo{year}{2014}\natexlab{}.
\newblock \showarticletitle{A type theory for productive coprogramming via
  guarded recursion}. In \bibinfo{booktitle}{\emph{CSL-LICS}}.
  \bibinfo{pages}{71:1--71:10}.
\newblock


\bibitem[\protect\citeauthoryear{M{\o}gelberg and Paviotti}{M{\o}gelberg and
  Paviotti}{2016}]%
        {GDTT:FPC}
\bibfield{author}{\bibinfo{person}{Rasmus~Ejlers M{\o}gelberg} {and}
  \bibinfo{person}{Marco Paviotti}.} \bibinfo{year}{2016}\natexlab{}.
\newblock \showarticletitle{Denotational semantics of recursive types in
  synthetic guarded domain theory}. In \bibinfo{booktitle}{\emph{Proceedings of
  the 31st Annual {ACM/IEEE} Symposium on Logic in Computer Science, {LICS}
  '16, New York, NY, USA, July 5-8, 2016}},
  \bibfield{editor}{\bibinfo{person}{Martin Grohe}, \bibinfo{person}{Eric
  Koskinen}, {and} \bibinfo{person}{Natarajan Shankar}} (Eds.).
  \bibinfo{publisher}{ACM}, \bibinfo{pages}{317--326}.
\newblock
\urldef\tempurl%
\url{https://doi.org/10.1145/2933575.2934516}
\showDOI{\tempurl}


\bibitem[\protect\citeauthoryear{Nakano}{Nakano}{2000}]%
        {Nakano:Modality}
\bibfield{author}{\bibinfo{person}{Hiroshi Nakano}.}
  \bibinfo{year}{2000}\natexlab{}.
\newblock \showarticletitle{A Modality for Recursion}. In
  \bibinfo{booktitle}{\emph{15th Annual {IEEE} Symposium on Logic in Computer
  Science, Santa Barbara, California, USA, June 26-29, 2000}}.
  \bibinfo{publisher}{{IEEE} Computer Society}, \bibinfo{pages}{255--266}.
\newblock
\urldef\tempurl%
\url{https://doi.org/10.1109/LICS.2000.855774}
\showDOI{\tempurl}


\bibitem[\protect\citeauthoryear{Orton and Pitts}{Orton and Pitts}{2016}]%
        {PittsAM:aximct}
\bibfield{author}{\bibinfo{person}{Ian Orton} {and} \bibinfo{person}{Andrew~M.
  Pitts}.} \bibinfo{year}{2016}\natexlab{}.
\newblock \showarticletitle{Axioms for Modelling Cubical Type Theory in a
  Topos}. In \bibinfo{booktitle}{\emph{25th {EACSL} Annual Conference on
  Computer Science Logic, {CSL} 2016, August 29 - September 1, 2016, Marseille,
  France}} \emph{(\bibinfo{series}{LIPIcs})},
  \bibfield{editor}{\bibinfo{person}{Jean{-}Marc Talbot} {and}
  \bibinfo{person}{Laurent Regnier}} (Eds.), Vol.~\bibinfo{volume}{62}.
  \bibinfo{publisher}{Schloss Dagstuhl - Leibniz-Zentrum fuer Informatik},
  \bibinfo{pages}{24:1--24:19}.
\newblock
\urldef\tempurl%
\url{https://doi.org/10.4230/LIPIcs.CSL.2016.24}
\showDOI{\tempurl}


\bibitem[\protect\citeauthoryear{Pous and Sangiorgi}{Pous and
  Sangiorgi}{2012}]%
        {pous2012enhancements}
\bibfield{author}{\bibinfo{person}{Damien Pous} {and} \bibinfo{person}{Davide
  Sangiorgi}.} \bibinfo{year}{2012}\natexlab{}.
\newblock \showarticletitle{{Enhancements of the bisimulation proof method}}.
\newblock In \bibinfo{booktitle}{\emph{{Advanced Topics in Bisimulation and
  Coinduction}}}, \bibfield{editor}{\bibinfo{person}{Davide Sangiorgi} {and}
  \bibinfo{person}{Jan Rutten}} (Eds.). \bibinfo{publisher}{{Cambridge
  University Press}}.
\newblock


\bibitem[\protect\citeauthoryear{Rutten}{Rutten}{2000}]%
        {Rut00}
\bibfield{author}{\bibinfo{person}{Jan J. M.~M. Rutten}.}
  \bibinfo{year}{2000}\natexlab{}.
\newblock \showarticletitle{Universal coalgebra: a theory of systems}.
\newblock \bibinfo{journal}{\emph{Theor. Comput. Sci.}} \bibinfo{volume}{249},
  \bibinfo{number}{1} (\bibinfo{year}{2000}), \bibinfo{pages}{3--80}.
\newblock
\urldef\tempurl%
\url{https://doi.org/10.1016/S0304-3975(00)00056-6}
\showDOI{\tempurl}


\bibitem[\protect\citeauthoryear{Schwencke}{Schwencke}{2010}]%
        {access}
\bibfield{author}{\bibinfo{person}{Daniel Schwencke}.}
  \bibinfo{year}{2010}\natexlab{}.
\newblock \showarticletitle{Coequational logic for accessible functors}.
\newblock \bibinfo{journal}{\emph{Inf. Comput.}} \bibinfo{volume}{208},
  \bibinfo{number}{12} (\bibinfo{year}{2010}), \bibinfo{pages}{1469--1489}.
\newblock
\urldef\tempurl%
\url{https://doi.org/10.1016/j.ic.2009.10.010}
\showDOI{\tempurl}


\bibitem[\protect\citeauthoryear{Spiwack and Coquand}{Spiwack and
  Coquand}{2010}]%
        {CS10}
\bibfield{author}{\bibinfo{person}{Arnaud Spiwack} {and}
  \bibinfo{person}{Thierry Coquand}.} \bibinfo{year}{2010}\natexlab{}.
\newblock \showarticletitle{{Constructively Finite?}}
\newblock In \bibinfo{booktitle}{\emph{{Contribuciones cient{\'i}ficas en honor
  de Mirian Andr{\'e}s G{\'o}mez}}},
  \bibfield{editor}{\bibinfo{person}{Laureano Lamb{\'a}n~Pardo},
  \bibinfo{person}{Ana Romero~Ib{\'a}{\~n}ez}, {and} \bibinfo{person}{Julio
  Rubio~Garc{\'i}a}} (Eds.). \bibinfo{publisher}{{Universidad de La Rioja}},
  \bibinfo{pages}{217--230}.
\newblock


\bibitem[\protect\citeauthoryear{Staton}{Staton}{2009}]%
        {Sta11}
\bibfield{author}{\bibinfo{person}{Sam Staton}.}
  \bibinfo{year}{2009}\natexlab{}.
\newblock \showarticletitle{Relating Coalgebraic Notions of Bisimulation}. In
  \bibinfo{booktitle}{\emph{Algebra and Coalgebra in Computer Science, Third
  International Conference, {CALCO} 2009, Udine, Italy, September 7-10, 2009.
  Proceedings}} \emph{(\bibinfo{series}{Lecture Notes in Computer Science})},
  \bibfield{editor}{\bibinfo{person}{Alexander Kurz}, \bibinfo{person}{Marina
  Lenisa}, {and} \bibinfo{person}{Andrzej Tarlecki}} (Eds.),
  Vol.~\bibinfo{volume}{5728}. \bibinfo{publisher}{Springer},
  \bibinfo{pages}{191--205}.
\newblock
\urldef\tempurl%
\url{https://doi.org/10.1007/978-3-642-03741-2\_14}
\showDOI{\tempurl}


\bibitem[\protect\citeauthoryear{{The Agda Team}}{{The Agda Team}}{2018}]%
        {Agda}
\bibfield{author}{\bibinfo{person}{{The Agda Team}}.}
  \bibinfo{year}{2018}\natexlab{}.
\newblock \bibinfo{title}{The Agda wiki}.
\newblock   (\bibinfo{year}{2018}).
\newblock
\newblock
\shownote{{\url{http://wiki.portal.chalmers.se/agda/}}.}


\bibitem[\protect\citeauthoryear{{The Project Everest Team}}{{The Project
  Everest Team}}{2018}]%
        {Everest}
\bibfield{author}{\bibinfo{person}{{The Project Everest Team}}.}
  \bibinfo{year}{2018}\natexlab{}.
\newblock \bibinfo{title}{The Everest Project}.
\newblock   (\bibinfo{year}{2018}).
\newblock
\newblock
\shownote{{\url{https://project-everest.github.io/}}.}


\bibitem[\protect\citeauthoryear{{Univalent Foundations Program}}{{Univalent
  Foundations Program}}{2013}]%
        {hott}
\bibfield{author}{\bibinfo{person}{The {Univalent Foundations Program}}.}
  \bibinfo{year}{2013}\natexlab{}.
\newblock \bibinfo{booktitle}{\emph{Homotopy Type Theory: Univalent Foundations
  of Mathematics}}.
\newblock \bibinfo{publisher}{\url{https://homotopytypetheory.org/book}},
  \bibinfo{address}{Institute for Advanced Study}.
\newblock


\bibitem[\protect\citeauthoryear{Veltri}{Veltri}{2017}]%
        {veltriphdthesis}
\bibfield{author}{\bibinfo{person}{Niccol{\`o} Veltri}.}
  \bibinfo{year}{2017}\natexlab{}.
\newblock \emph{\bibinfo{title}{A Type-Theoretical Study of Nontermination}}.
\newblock \bibinfo{thesistype}{Ph.D. Dissertation}. \bibinfo{school}{Tallinn
  University of Technology}.
\newblock
\urldef\tempurl%
\url{https://digi.lib.ttu.ee/i/?7631}
\showURL{%
\tempurl}


\bibitem[\protect\citeauthoryear{Vezzosi}{Vezzosi}{2017}]%
        {vezzosi2017streams}
\bibfield{author}{\bibinfo{person}{Andrea Vezzosi}.}
  \bibinfo{year}{2017}\natexlab{}.
\newblock \bibinfo{title}{Streams for Cubical Type Theory}.
  (\bibinfo{year}{2017}).
\newblock
\urldef\tempurl%
\url{http://www.cse.chalmers.se/~vezzosi/streams-ctt.pdf}
\showURL{%
\tempurl}


\end{thebibliography}

\end{document}